\PassOptionsToPackage{prologue,dvipsnames}{xcolor}
\documentclass[11pt]{article}

\usepackage{comment}
\usepackage{lmodern}
\usepackage{amsmath}
\usepackage{amssymb}
\usepackage{mathtools}
\usepackage{amsthm}
\usepackage{mdframed}

\usepackage{booktabs} 

\usepackage[ruled]{algorithm2e} 

\SetKw{Continue}{continue} %
\SetAlFnt{\small}
\SetAlCapFnt{\small}
\SetAlCapNameFnt{\small}
\SetAlCapHSkip{0pt}

\IncMargin{-\parindent}
\usepackage[square]{natbib}

\setcitestyle{authoryear}
\usepackage{authblk}
\usepackage{titlesec}

\title{\begin{flushleft}\LARGE\textbf{Fairness in the Multi-Secretary Problem}\end{flushleft}}

\author{\begin{flushleft}\large 
GEORGIOS PAPASOTIROPOULOS, \textit{\small University of Warsaw}\\\vspace{-0.2cm}
ZEIN PISHBIN, \textit{\small Université Paris Dauphine‑PSL}
	\end{flushleft}
}

\usepackage{geometry}

\usepackage{tikz}
\usetikzlibrary{positioning,patterns,decorations.pathreplacing,arrows.meta,shapes.geometric,calc}

\usepackage{nicefrac}
\usepackage{paralist}
\usepackage{enumitem}
\usepackage{xcolor}
\usepackage{colortbl}
\usepackage{subcaption}
\usepackage{listings}
\usepackage{fontawesome}
\usepackage{multirow}
\usepackage{tabularx}
\usepackage{pgfplots}
\pgfplotsset{compat=1.18}
\usepackage{pifont}

\definecolor{winered}{rgb}{0.5,0.1,0.1}
\definecolor{darkgreen}{rgb}{0,0.35,0}

\newcolumntype{Y}{>{\centering\arraybackslash}X}
\newcolumntype{Z}{>{\raggedleft\arraybackslash}X}
\newcolumntype{a}{>{\columncolor{black!05}}Y}

\renewcommand*{\leq}{\leqslant}

\renewcommand*{\geq}{\geqslant}
\renewcommand{\epsilon}{\varepsilon}

\renewcommand{\part}{{{\mathrm{part}}}}

\usepackage{thmtools,thm-restate}

\usepackage{setspace}

\renewcommand*{\leq}{\leqslant}

\renewcommand*{\geq}{\geqslant}
\renewcommand{\epsilon}{\varepsilon}

\usepackage[T1]{fontenc}

\renewcommand{\part}{{{\mathrm{part}}}}

\newcommand{\pabulib}{{{\textsc{Pabulib}}}}

\newcommand{\lv}[1]{}

\usepackage{etoolbox}
\usepackage{tablefootnote}

\usepackage{titlesec}

\usepackage{tocloft}
\usepackage{multicol}

\makeatletter
\AtBeginDocument{%
	\begingroup
	\normalsize
	\let\tmp@n@s\f@size
	\let\tmp@n@b\f@baselineskip
	\small
	\let\tmp@s@s\f@size
	\let\tmp@s@b\f@baselineskip
	\xdef\semismall@size{\fpeval{(\tmp@n@s+\tmp@s@s)/2}}%
	\xdef\semismall@baselineskip{\fpeval{(\tmp@n@b+\tmp@s@b)/2}}%
	\endgroup
}
\newcommand{\semismall}{\fontsize{\semismall@size}{\semismall@baselineskip}\selectfont}

\usepackage{multirow}
\usepackage{arydshln}
  \usepackage{subcaption}
\usepackage{thmtools}
\usepackage{thm-restate}
\usepackage{mathtools}

\usepackage[ruled]{algorithm2e} 
\definecolor{verylightgray}{gray}{0.9}
\definecolor{crimson}{RGB}{220, 20, 60}
\usepackage{colortbl}

\usetikzlibrary{decorations.pathreplacing}

\allowdisplaybreaks

\usepackage{pifont}
\definecolor{winered}{rgb}{0.5,0.1,0.1}

\renewcommand\emph[1]{{\color{winered}{{\textit{#1}}}}}

\PassOptionsToPackage{hyphens}{url}
\usepackage[colorlinks=true, linkcolor=winered, citecolor=darkgreen, urlcolor=blue]{hyperref}
\usepackage[sort&compress,nameinlink]{cleveref}
\Crefname{table}{Table}{Tables}
\Crefname{figure}{Figure}{Figures}
\Crefname{theorem}{Theorem}{Theorems}
\Crefname{definition}{Definition}{Definitions}
\Crefname{corollary}{Corollary}{Corollaries}
\Crefname{observation}{Observation}{Observations}
\Crefname{question}{Question}{Question}
\Crefname{lemma}{Lemma}{Lemmas}
\Crefname{example}{Example}{Examples}
\Crefname{reduction}{Reduction}{Reductions}
\Crefname{construction}{Construction}{Constructions}
\Crefname{subsection}{Section}{Sections}
\Crefname{section}{Section}{Sections}
\Crefname{proposition}{Proposition}{Propositions}
\Crefname{algorithm}{Algorithm}{Algorithms}
\Crefname{algocf}{Algorithm}{Algorithms}
\Crefname{lstlisting}{listing}{listings}
\Crefname{equation}{Relation}{Relations}
\Crefname{claim}{Claim}{Claims}

\newtheoremstyle{ex}
{}{}
{}
{}
{\bfseries}
{.}
{ }
{%
	\thmname{#1}
	\thmnumber{ #2}
	\thmnote{{\bfseries: \ifex(\fi#3\ifex) \fi}}
}
\newif\ifex

{\theoremstyle{ex}{
		}}
\Crefname{exof}{Example}{Examples}

\theoremstyle{definition}
\newtheorem{definition}{Definition}
\newtheorem{example}{Example}

\theoremstyle{plain}
\newtheorem{theorem}{Theorem}

\usepackage{utopia}

\begin{document}

	\date{}

\definecolor{mygreen}{RGB}{130,230,130}  
\definecolor{myred}{RGB}{255,150,150}
\newcommand{\greenorange}[1]{%
  \begin{tikzpicture}[baseline=(current bounding box.center)]
    \node[minimum width=2em, minimum height=1.5em, inner sep=0pt,
          fill=mygreen, 
          path picture={
            \fill[myred] (path picture bounding box.north west) -- 
                         (path picture bounding box.south east) -- 
                         (path picture bounding box.north east) -- cycle;
          }] {#1};
  \end{tikzpicture}%
}

\newcommand{\greenblack}[1]{%
  \begin{tikzpicture}[baseline=(current bounding box.center)]
    \node[minimum width=2em, minimum height=1.5em, inner sep=0pt,
          fill=mygreen, 
          path picture={
            \fill[black!40] (path picture bounding box.north west) -- 
                               (path picture bounding box.south east) -- 
                               (path picture bounding box.south west) -- cycle;
          }] {#1};
  \end{tikzpicture}%
}

\newcommand{\orangeblack}[1]{%
  \begin{tikzpicture}[baseline=(current bounding box.center)]
    \node[minimum width=2em, minimum height=1.5em, inner sep=0pt,
          fill=myred, 
          path picture={
            \fill[black!40] (path picture bounding box.north west) -- 
                               (path picture bounding box.south east) -- 
                               (path picture bounding box.south west) -- cycle;
          }] {#1};
  \end{tikzpicture}%
}

\definecolor{myblue}{RGB}{130,200,255}
\newcommand{\solidblue}[1]{%
  \begin{tikzpicture}[baseline=(current bounding box.center)]
    \node[minimum width=2em, minimum height=1.5em, inner sep=0pt,
          fill=myblue] {#1};
  \end{tikzpicture}%
}

\newcommand{\solidwhite}[1]{%
  \begin{tikzpicture}[baseline=(current bounding box.center)]
    \node[minimum width=2em, minimum height=1.5em, inner sep=0pt,
          fill=white!40] {#1};
  \end{tikzpicture}%
}

	\newgeometry{left=0.95in, right=0.95in, top=-0.05in, bottom=0.8in}
	\maketitle

	\vspace{-0.5cm}
    {\small\noindent
This paper bridges two perspectives: it studies the multi-secretary problem through the fairness lens of social choice, and examines multi-winner elections from the viewpoint of online decision making. After identifying the limitations of the prominent proportionality notion of Extended Justified Representation (EJR) in the online domain, the work proposes a set of mechanisms that merge techniques from online algorithms with rules from social choice---such as the Method of Equal Shares and the Nash Rule---and supports them through both theoretical analysis and extensive experimental evaluation.}

	\setcounter{tocdepth}{2} 
	\vspace{0.8cm}\noindent\rule{\textwidth}{0.8pt}
	\begin{center}
		\textsc{Contents}  
	\end{center}
	\vspace{-1.2cm}
	\renewcommand{\contentsname}{}
	\begin{spacing}{0.9}
		\tableofcontents 
	\end{spacing}
	\noindent\rule{\textwidth}{0.5pt}
	
	\setcounter{secnumdepth}{2}
	
	\newgeometry{left=0.95in, right=0.95in, top=0.75in, bottom=0.8in}

\section{Introduction}
\label{sec:intro}
The \emph{secretary problem}
appeared in the 1950s and quickly became a central topic in probability theory, optimization, and algorithms \citep{freeman1983secretary,ferguson1989solved,samuels1991secretary,dynkin1963optimum}. It is a simple and elegant model that led to the development of an entire field around online and sequential decision-making processes. In its base form, a known number of candidates arrive one by one in random order. A single decision maker must either accept or reject a candidate immediately upon their arrival. Once rejected, a candidate cannot be recalled, and once accepted, they cannot later be rejected or exchanged for another one. At the moment of arrival of a candidate, the decision maker gathers enough information to evaluate them with respect to those seen so far, but remains unaware of the quality of the yet unseen ones. The goal is to pick the best candidate overall. 

Often, real-world scenarios depart from this setting. 
It is typical to \emph{select multiple candidates} instead of one, and for \emph{decisions to be made by electorates}, rather than individuals, whose members may have varying and conflicting preferences.
Then, \emph{fairness among decision-makers} becomes crucial to ensure that diverse interests are adequately represented. 
In our work, we extend the base secretary model to reflect these aspects, still under the constraint that decisions must be made sequentially and irrevocably as candidates arrive. 
The following example illustrates our fairness considerations.

\begin{example}
    \label{example1} 
\textit{A hiring board in a Mathematics department is composed equally of theoreticians and applied mathematicians. Each group prefers candidates from their own field.
Due to the heavy workload involved in evaluating individual applications, the members of the board decide not to review all candidates at once but to use a sequential process. Since strong candidates may accept other offers if kept waiting, the board aims to decide on each candidate immediately upon evaluation.
The task is to select $k = 2$ candidates from a pool of six applicants ${c_1, c_2, \dots, c_6}$ considered in the order of their arrival (i.e., increasing index). The preferences of the two groups within the hiring board are expressed through cardinal ballots, as shown below; the colors indicate different candidate selections, which will be discussed shortly.}

\begin{table}[h!]
\centering
\small
\begin{tabular}{@{}c|cccccc@{}}
    \toprule
    & \solidblue{$c_1$} & \solidblue{$c_2$} & \greenorange{$c_3$} & \greenblack{$c_4$} & \solidwhite{$c_5$} & \orangeblack{$c_6$} \\ 
    \midrule
    \multicolumn{1}{l|}{Theory} & $0$  & $1$ & $2$ & $0$ & $0$ & $0$ \\
    \multicolumn{1}{l|}{Applied} & $2$ & $0$ & $0$ & $3$ & $1$ & $3$ \\ 
    \bottomrule
    \end{tabular}
\end{table}
\textit{The solution that consists of the candidates highlighted in black, namely $\{c_4, c_6\}$, clearly achieves the highest total utility.
    However, it leaves half the electorate (the theoretical mathematicians) completely unsatisfied. A fairer solution is to hire one candidate from the applied field ($c_1, c_4, c_5, c_6$) and one from theory ($c_2, c_3$).
    The three rules proposed in our work, namely \emph{Greedy Budgeting}, \emph{Online Method of Equal Shares} and \emph{Online Nash Product}, would take the online element into account and select the (blue) solution $\{c_1, c_2\}$, the (green) solution $\{c_3, c_4\}$, and the (red) solution $\{c_3, c_6\}$, respectively. While these sacrifice some overall utility, they ensure that no field is left unrepresented.}
\end{example}

Our work is shaped by three core dimensions: (i) \emph{online arrival of candidates}, (ii) \emph{fairness goal}, (iii) \emph{cardinal ballots of voters}.
This is the first paper to bring all these elements together.
\emph{Fairness under cardinal ballots} in offline settings has been extensively explored in computational social choice \citep{benade2023participatory,peters2021proportional,fain2016core,bhaskar2018truthful}; when fairness is not a goal and total utility is considered in an online fashion, the problem reduces to the \emph{multi-secretary problem} \citep{kleinberg2005multiple,besbes2022multi}; and when ballots are binary (i.e., approval-based), the problem corresponds to \emph{online multi-winner elections} \citep{Do}.
Our study lies at the intersection of online algorithms and social choice.

\begin{itemize}[left=0em, topsep=0pt, itemsep=0pt]
    \item \textit{It explicitly adapts concepts and rules from voting theory} (such as the Extended Justified Representation axiom, the Method of Equal Shares, the Method of Equal Shares with Bounded Overspending, the Nash Product rule) 
    \textit{to address fairness considerations in classical online 
    selection problems} (i.e., in the (multi-)secretary problem), and
    \item \textit{It applies algorithmic techniques from the online algorithms literature} (specifically, two distinct optimal stopping policies) \textit{to address foundational problems in social choice theory} (namely, multi-winner/committee elections with cardinal ballots) \textit{when extended to dynamic settings}.
\end{itemize}

\noindent Therefore, it aligns with two lines of work.
\emph{Dynamic social decision frameworks} have recently attracted considerable attention in computational social choice, with fairness considerations often being the central focus; see, indicatively, the works by \citet{kehne2025robust,dong2024proportional,halpern2023representation,hossain2021fair}, or the survey by \citet{elkind2024temporal}.
Yet, only \citet{Do} consider an online arrival of candidates requiring immediate decisions, akin to secretary problems.
The study of \emph{fairness in online algorithms} has also developed recently.
\citet{balkanski2024fair} study fairness for variants of the secretary problem in learning-augmented settings. 
Fairness in the single-candidate selection model is the focus of \citet{arsenis2022individual}.
Another notable work is that of \citet{correa2021fairness}, which is conceptually close to ours but also limited to selecting one candidate; it explicitly states the multi-selection setting as an open problem---which we address.\looseness-1

\paragraph*{Contributions.}
Our work can be seen as extending the study of \citet{Do} by considering cardinal ballots. \citeauthor{Do} showed that moving from the offline to the online setting entails only a modest loss in fairness: although the prominent notion of Extended Justified Representation (EJR)---which is achievable offline---cannot be fully satisfied online, it can still be well-approximated.
This was somewhat surprising, as one might reasonably expect a more significant drop in quality of the outcome when the selection mechanism lacks knowledge of the entire input.
We show that this counterintuitive conclusion no longer holds when voters submit more expressive ballots than simple approvals. Specifically, we analyze three natural relaxations of EJR and prove that none of them can be reasonably approximated in the setting we examine.
In response to these negative observations, we explore alternative paths to achieving fairness guarantees online. We propose three distinct rules
and we evaluate their performance both theoretically and experimentally. 
One of the rules we define, namely Online Method of Equal Shares, can be seen as fitting a general framework we introduce that converts offline multi-winner election rules into online ones.
Based on this observation we additionally formalize an online variant of the recently introduced Method of Equal Shares with Bounded Overspending \citep{papasotiropoulos2024method}, and show that it consistently stands out due to its strong performance.
Finally, in the course of formalizing fairness, we also introduce the angle of probabilistic satisfiability of known proportionality notions: a conceptual contribution that may be of interest to social choice beyond the online domain.

\section{Preliminaries}

We adopt the terminology and notation of social choice theory
as the fairness concepts and selection methods we explore are mostly inspired by it. 
An \emph{election} is a tuple $E = (C, V, k, \{u_i\}_{i\in [n]})$: 

    \begin{itemize}[left=0em, topsep=0pt, itemsep=0pt]
    \item $C = \{1, 2, \dots, m\}$ and $V = \{1, 2, \dots, n\}$ are the sets of \emph{candidates} and \emph{voters} respectively, \item $k$ is the upper bound on the \emph{size of the committee} to be selected, and we assume $2\leq k< m$,\item for each voter $i\in V$, $u_i: C\mapsto \mathbb{R}_{\geq 0}$ is a function that corresponds to the \emph{cardinal ballot} that voter $i$ casts indicating their \emph{additive utility}. If a set $W \subseteq C$ is elected, the overall utility of voter $i$ is $u_i(W)=\sum_{j\in W}u_i(j)$. 
\end{itemize} 

If $C$ is an ordered set, meaning the candidates are presented in a specific sequence that determines their appearance order, such that for each $i\in [n]$ the value of $u_i(j)$ is revealed at the moment when the candidate $j$ appears, then $E$ is an \emph{online election}. Otherwise, it is an \emph{offline election}.
In case $u_i(j)\in \{0,1\}, \forall (i,j)\in [n]\times[m],$ we say that we are in the approval setting and voters are expressing \emph{approval preferences}.
The definition of cardinal ballots above allows for arbitrary scores. A more realistic variant, which already generalizes approval preferences, involves \emph{range ballots}. There, voters rate each candidate on a fixed scale—such as with integers from 0 to 5, 10 or 20 (as done in online reviews, grading, employee evaluations, or satisfaction surveys). Then, we say that the considered election $E$ comes with an additional integer parameter $B$ and voters' ballots satisfy $u_i(j)\leq B$.

A subset of candidates $W \subseteq C$ is a \emph{feasible solution} if $|W|\leq k$. We denote the set of all feasible solutions by $\mathcal{W}$.
Our goal is to choose a feasible subset of
candidates, which we call an \emph{outcome}, based on voters’ ballots. 
A \emph{multi-winner election rule} $\mathcal{R}$, in short, a rule, is a function that takes as input an election $E=(C, V, k, \{u_i\}_{i\in [n]})$ and returns an outcome $\mathcal{R}(E)$ called the \emph{winning committee}.
In particular, an \emph{online multi-winner election rule} is an algorithm that processes candidates of an online election in the predefined sequence given by $C$, making an irrevocable decision at each step whether to include the currently examined candidate in the winning committee or not. If included, the candidate is said to be \emph{hired/elected}; otherwise, they are \emph{rejected}.
When making such a decision on candidate $j \in [m]$, the rule can only use information about voters' preferences for candidates that have already appeared, i.e., $u_i(c), \forall c\leq j$ and
$i\in [n]$.

\subsection{Fairness Considerations from Social Choice}
\label{sec:prelims2}
A prominent concept in achieving fair representation is the axiom of Extended Justified Representation (EJR) \citep{aziz2017justified}. It has received significant attention in the literature when it comes to formalizing fairness---especially in the context of multi-winner elections and
participatory budgeting (PB): the settings most closely
related to our framework---and its influence extends even beyond traditional voting  \citep{kellerhals2024proportional,papasotiropoulos2025proportional,boehmer2024approval}. As such, EJR plays a central role
in our work. It applies directly to both offline and online elections. 
Further proportionality axioms, which are implied by EJR, will also be presented and examined later, in \Cref{sec:JR}.

\begin{definition}
        A group of voters $S$ is \emph{$(\alpha, T)$-cohesive}, where $\alpha: C\mapsto
\mathbb{R}_{\geq0}$ and $T\subseteq C$, if $\frac{|S|}{n} \geq \frac{|T|}{k}$ and $u_i(c) \geq \alpha(c)$ for all $i\in S$ and $c\in T$. 
A rule $\mathcal{R}$ satisfies \emph{Extended Justified Representation} (\emph{EJR}) if for each election instance $E$, and each
$(\alpha, T)$-cohesive group of voters $S$ there exists a voter $i\in S$ such that $u_i(\mathcal{R}(E))\geq \sum_{c\in T} \alpha(c)$.
\end{definition}

Among the election rules examined in computational social choice from the perspective of proportionality, the \emph{Method of Equal Shares} (\emph{MES}),
which satisfies EJR, has received particular attention and has also been applied in real-life (offline) elections.
For its formal definition we refer the reader to the work of \cite{peters2021proportional}. 
Apart from EJR (and variants of it), we also explore an additional fairness concept in \Cref{sec:nash}, namely Nash Welfare Maximization.

\section{Proportional Selection with Certainty}
\label{sec:certainty}
As mentioned, EJR has become the most well-studied fairness notion in voting literature. Therefore, we first focus on EJR, along with its approximations and relaxations, to explore what is impossible and what can be guaranteed in this regard, in the online setting. In parallel, we aim to understand how these results differ from those by \citet{Do}, where voters were limited to expressing only approval preferences.

\subsection{Satisfying and Approximating Extended Justified Representation}
\label{sec:ejr}

It is already known that EJR cannot be satisfied by an online rule \citep{Do}.
Consequently, in hopes of positive results, we turn to approximate notions of EJR instead, by exploring three distinct ways to get around this limitation.
 The first variant we introduce aims to guarantee that each group of voters receives at least a fraction of the utility they would deserve according to EJR and it is simply inspired by multiplicative guarantees of approximation algorithms. The second notion closely relates to an established relaxation of EJR, ensuring that every group of voters would be appropriately represented if additional candidates were allowed in the outcome \citep{peters2021proportional}.
 The third variant ensures fairness for sufficiently large groups of voters, specifically those that exceed the size required by EJR by a fixed factor---this is the one explored by \citet{Do}.

\begin{definition}
\label{def:ejr_approx}
    A rule $\mathcal{R}$ satisfies:
    \begin{itemize}[left=0em, topsep=0pt, itemsep=0pt]
        \item \emph{$\beta$-approximation of Extended Justified Representation} (\emph{$\beta$-EJR}), for $\beta\geq 1,$ if for each election instance $E$, and each
$(\alpha, T)$-cohesive group of voters $S$ there exists a voter $i\in S$ such that $u_i(\mathcal{R}(E))\geq \nicefrac{1}{\beta}\sum_{c\in T} \alpha(c)$.
\item \emph{Extended Justified Representation up to $\gamma$ candidates} (\emph{EJR-$\gamma$}), for $\gamma\leq k-1$, if for each election instance $E$, and each 
$(\alpha, T)$-cohesive group of voters $S$ there exists a voter $i\in S$ and a set of candidates $X$ where $|X|\leq \gamma,$ such that $u_i(\mathcal{R}(E)\cup X)\geq \sum_{c\in T} \alpha(c)$.
\item \emph{Extended Justified Representation for $\delta$-cohesive groups}, for $\delta\leq k,$ 
if for each election instance $E$, and each 
$(\alpha, T)$-cohesive group of voters $S$ satisfying $\frac{|S|}{n} \geq \delta \frac{|T|}{k}$
there exists a voter $i\in S$ such that $u_i(\mathcal{R}(E))\geq \sum_{c\in T} \alpha(c)$. 
    \end{itemize}
\end{definition}

Below we show that in the examined framework, not only EJR cannot be satisfied, 
but it also cannot be reasonably approximated under any of the three approximate notions defined. 
The constructed instances highlight that the impossibility comes exactly because of the online nature of the examined problem combined with cardinal ballots and they show a notable contrast to the approval setting, where EJR could be well approximated \citep{Do}.

\begin{theorem}
\label{thm:ejr1}
    There is no online voting rule that satisfies $\beta$-EJR, for any finite positive value of $\beta$, even for range ballots.
\end{theorem}
\begin{proof}
Fix an arbitrary $\beta,$ some $k$ and an arbitrarily small positive value $\epsilon$. Consider an instance on $2k$ candidates, namely $\{a_1,a_2,\dots,a_k,b_1,\dots, b_k\}$ and $k$ voters, the cardinal ballots of which appear in the table below and are upper bounded by $B=\beta$. 

\begin{table}[h!]
\centering
\begin{tabular}{@{}c|ccccccc@{}}
\toprule
                              & $a_1$        & $a_2$        & $\dots$ & $a_k$        & $b_1$             & $\dots$ & $b_k$             \\ \midrule
\multicolumn{1}{l|}{$v_1$}   & $1-\epsilon$ &              &         &              & $\nicefrac{\beta}{k}$ & $\dots$ & $\nicefrac{\beta}{k}$ \\
\multicolumn{1}{l|}{$v_2$}   &              & $1-\epsilon$ &         &              & $\nicefrac{\beta}{k}$ & $\dots$ & $\nicefrac{\beta}{k}$ \\
\multicolumn{1}{l|}{$\dots$} &              &              &         &              & $\dots$           &         &                   \\
\multicolumn{1}{l|}{$v_k$}   &              &              &         & $1-\epsilon$ & $\nicefrac{\beta}{k}$ & $\dots$ & $\nicefrac{\beta}{k}$ \\ \bottomrule
\end{tabular}
\end{table}

Say that the candidates are arriving in the order presented, i.e. first the $a$-candidates and then the $b$-candidates, all in increasing index order.
Candidate $a_1$ appears first. Say that $T=\{a_1\}$ and $S=\{v_1\}.$ For these groups it holds that $|S|\geq |T|\nicefrac{n}{k},$ and $\alpha(c)=1-\epsilon, \forall c\in T$. According to $\beta$-EJR, at least one voter in $S$ should receive a satisfaction of at least $\frac{1-\epsilon}{\beta}$ from the outcome.
The decision on whether to include $a_1$ in the winning committee or not should be done at this point and any online multi-winner election rule that rejects $a_1$ can't revoke this decision in case the candidates that will come later do not satisfy $v_1$ to an extent of $1-\epsilon$. Therefore, any rule that does not have a knowledge of the preferences of $v_1$ for the candidates that are yet to come, should include $a_1$ in the committee at this point. 
By the same argument, all candidates in the set $\{a_1,a_2,\dots,a_k\}$ will be selected one after the other. Hence, there is no place to include any candidate from $\{b_1,\dots,b_k\}$ in the solution. 
It holds that any rule that does not elect any candidate from $\{b_1,\dots,b_k\}$ does not satisfy $\beta$-EJR. 
This is because if we set $T=\{b_1,\dots,b_k\}$ and $S=\{v_1,v_2,\dots,v_k\},$ then $S$ is $((\frac{\beta}{k})_{i\in [k]},T)$-cohesive. Hence, in order to satisfy $\beta$-EJR, there should be a voter $i$ in $S$ such that 
$u_i(\mathcal{R}(E))\geq \frac{1}{\beta}\sum_{c\in T} \frac{\beta}{k}=1.$ But, under a rule that does not elect any candidate from $\{b_1,\dots,b_k\}$,
each voter in $S$ gets a satisfaction of at most $1-\epsilon$.
\end{proof}

The axiom of EJR up to $\gamma$ candidates is also unsatisfiable.
The proof follows the same rationale as that of \Cref{thm:ejr1}, though it employs a somewhat more involved construction.

\begin{theorem}
    There is no online voting rule that satisfies EJR-$\gamma,$ for any integer $\gamma\leq k-1$, even for range ballots
    \label{thm:ejr2}
\end{theorem}

\begin{proof}
Fix an arbitrary $\gamma,$ some $k$ and an arbitrarily small positive value $\epsilon$. Consider an instance on $k^2+k$ candidates, namely $\{a_1^1,a_1^2, \dots, a_1^k,
a_2^1,a_2^2, \dots, a_2^k,\dots,
a_k^1,a_k^2, \dots, a_k^k,
 b_1, \dots, b_k\}$ and $k$ voters, the cardinal ballots of which appear in the table below and are upper bounded by $B=2^{k+1}\epsilon$.

\begin{table*}[h!]
\centering
\begin{tabular}{@{}c|ccccccccccccccccc@{}}
\toprule
                              & $a_1^1$    & $a_1^2$     & $\dots$ & $a_1^k$       & $a_2^1$    & $a_2^2$     & $\dots$ & $a_2^k$       & $\dots$ & $a_k^1$    & $a_k^2$     & $\dots$ & $a_k^k$       & $b_1$             & $\dots$ & $b_k$             \\ \midrule
\multicolumn{1}{l|}{$v_1$}   & $\epsilon$ & $2\epsilon$ & $\dots$ & $2^k\epsilon$ &            &             &         &               &         &            &             &         &               & $2^{k+1}\epsilon$ &         & $2^{k+1}\epsilon$ \\
\multicolumn{1}{l|}{$v_2$}   &            &             &         &               & $\epsilon$ & $2\epsilon$ & $\dots$ & $2^k\epsilon$ &         &            &             &         &               & $2^{k+1}\epsilon$ &         & $2^{k+1}\epsilon$ \\
\multicolumn{1}{l|}{$\dots$} &            &             &         &               &            &             &         &               &         &            &             &         &               &                   &         &                   \\
\multicolumn{1}{l|}{$v_k$}   &            &             &         &               &            &             &         &               &         & $\epsilon$ & $2\epsilon$ & $\dots$ & $2^k\epsilon$ & $2^{k+1}\epsilon$ &         & $2^{k+1}\epsilon$ \\ \bottomrule
\end{tabular}
\end{table*}

Say that the candidates are arriving in the order presented, i.e. first all the $a_1$-candidates, then all the $a_2$-candidates, etc, then the all $a_k$-candidates and then the $b$-candidates, all in increasing order of superscript.
First, we note that at least one candidate from $\{\alpha_1^1,\dots,\alpha_1^k\}$ will be included in the solution by the time that $\alpha_1^k$ has been examined. This is because if at that point we consider $T=\{a_1^k\}$ and $S=\{v_1\},$ it holds that $|S|\geq |T|\nicefrac{n}{k},$ and $\alpha(c)=2^k\epsilon, \forall c\in T$.
According to EJR-$\gamma$, at least one voter in $S$ should receive a satisfaction of at least $2^k\epsilon$ from $\mathcal{R}(E)\cup X$, for some $X$ where $|X|\leq k-1$.
This is because any online multi-winner election rule that rejects $a_1^k$ can't revoke this decision in case the candidates that will come are not approved by $v_1$, and then, if $X=\{a_1^1,a_1^2,\dots,a_1^{k-1}\}$ it holds $u_i(\mathcal{R}(E)\cup X)=\epsilon\sum_{i=0}^{k-1}2^i=\epsilon(2^k-1)=2^k\epsilon-\epsilon$. Therefore, $a_1^{k}$ should be selected at the time it arrives. By similar arguments, after examining all $a$-candidates, it should hold that at least one candidate approved by each voter should be selected. Hence, there is no place to include any candidate from $\{b_1,\dots,b_k\}$ in the solution. 

It holds that any rule that does not elect any candidate from $\{b_1,\dots,b_k\}$ does not satisfy EJR-$\gamma$. This is because if we set $T=\{b_1,\dots,b_k\}$ and $S=\{v_1,v_2,\dots,v_k\},$ then $S$ is $((2^{k+1}\epsilon)_{i\in [k]},T)$-cohesive. Hence, in order to satisfy EJR-$\gamma$, there should be a voter $i$ in $S$ such that 
$u_i(\mathcal{R}(E)\cup X)\geq \sum_{c\in T} \alpha(c)=\epsilon k2^{k+1}.$
However, for each voter $i$ in $S$, using the 
set $X$ that consists of $(k-1)$ $b$-candidates, it holds that
$u_i(\mathcal{R}(E)\cup X) \leq 2^k\epsilon+ \sum_{i\in [k-1]}\epsilon 2^{k+1}=2^k\epsilon+\epsilon (k-1) 2^{k+1}= 2^{k}\epsilon(1+2(k-1))=2^{k}\epsilon(2k-1))<2^{k}\epsilon2k$. 
\end{proof}

Finally, akin to the previous relaxations of EJR, the third one introduced in \Cref{def:ejr_approx} is also impossible to satisfy.

\begin{theorem}
\label{thm:ejr3}
    There is no online voting rule that satisfies EJR for $\delta$-cohesive groups, for any $\delta \leq k$, even for range ballots.
\end{theorem}
\begin{proof}
Fix an arbitrary $\delta,$ some $k$ and an arbitrarily small positive value $\epsilon$. Consider an instance on $2k$ candidates, namely $\{a_1,a_2,\dots,a_k,b_1,\dots, b_k\}$ and $1$ voter, the cardinal ballot of which appears in the following table and are upper bounded by $B=1+k\epsilon$.
\begin{table*}[h!]
\centering
\begin{tabular}{@{}c|ccccccc@{}}
\toprule
                            & $a_1$        & $a_2$         & $\dots$ & $a_k$         & $b_1$                     & $\dots$ & $b_k$                     \\ \midrule
\multicolumn{1}{l|}{$v_1$} & $1+\epsilon$ & $1+2\epsilon$ &         & $1+k\epsilon$ & $1+\epsilon(\frac{k+1}{2}+\frac{1}{k})$ & $\dots$ & $1+\epsilon(\frac{k+1}{2}+\frac{1}{k})$ \\ \bottomrule
\end{tabular}
\end{table*}

Say that the candidates are arriving in the order presented, i.e. first the $a$-candidates and then the $b$-candidates, all in increasing index order.
Candidate $a_1$ appears first. Say that $T=\{a_1\}$ and $S=\{v_1\}.$ For these groups it holds that $|S|\geq |T|\nicefrac{n}{k},$ and $\alpha(c)=1+\epsilon, \forall c\in T$. According to EJR for $\delta$-cohesive groups, at least one voter in $S$ should receive a satisfaction of at least $1+\epsilon$ from the outcome.
The decision on whether to include $a_1$ in the winning committee or not should be done at this point and any online multi-winner election rule that rejects $a_1$ can't revoke this decision in case the candidates that will come do not satisfy $v_1$ to an extent of $1+\epsilon$. Therefore, any rule that does not have a knowledge of the preferences of $v_1$ for the candidates that are yet to come, should include $a_1$ in the committee at this point. 
Similarly, when $a_2$ arrives, we have that $v_1$ should now receive a satisfaction of at least $1+2\epsilon$ from the outcome. Currently, having included $a_1$, the satisfaction of $v_1$ is $1+\epsilon,$ so not sufficiently large to satisfy EJR for $\delta$-cohesive groups in case the candidates that will come next do not satisfy $v_1$, and, hence $a_2$ should be included in the committee as well. By similar arguments, all candidates in the set $\{a_1,a_2,\dots,a_k\}$ will be selected one after the other. Hence, there is no place to include any candidate from $\{b_1,\dots,b_k\}$ in the solution. 

It holds that any rule that does not elect any candidate from $\{b_1,\dots,b_k\}$ does not satisfy EJR for $\delta$-cohesive groups. This is because if we set $T=\{b_1,\dots,b_k\}$ and $S=\{v_1\},$ then $S$ is $((1+\epsilon(\frac{k+1}{2}+\frac{1}{k}))_{i\in [k]},T)$-cohesive. Hence, in order to satisfy EJR for $\delta$-cohesive groups, there should be a voter $i$ in $S$ such that 
$u_i(\mathcal{R}(E))\geq \sum_{c\in T} (1+\epsilon(\frac{k+1}{2}+\frac{1}{k}))=k+\epsilon\frac{k(k+1)}{2}+\epsilon.$ But, under a rule that does not elect any candidate from $\{b_1,\dots,b_k\}$,
each voter in $S$ gets a satisfaction of at most $\sum_{i\in [k]}1+i\epsilon=k+\epsilon\frac{k(k+1)}{2}$.
\end{proof}

Observe that the negative results of our work do not rely on complexity assumptions, demonstrating that the guarantees we aim for as per \Cref{def:ejr_approx} are impossible to achieve under any rule regardless of its running time. On the other hand, naturally, when presenting positive results we will focus exclusively on polynomial-time rules.
While one could explore combinations of the approximate notions we defined, due to the strongly negative results of the preceding theorems it seems more prudent to move away from EJR. This is what we do in \Cref{sec:JR} and \Cref{sec:probable}.

\subsection{Satisfying Weakenings of Extended Justified Representation}
\label{sec:JR}
Proportional Justified Representation (PJR) and Justified Representation (JR) are two well studied relaxations of EJR. For the definition of PJR, we refer to the work of \citet{sanchez2017proportional}, and we note that natural approximate notions of it can be defined in accordance to \Cref{def:ejr_approx}. Importantly, all the negative results presented in \Cref{thm:ejr1,thm:ejr2,thm:ejr3} continue to hold for PJR as well---for the latter, this follows directly from its proof; for the others, it follows after dividing voters' utilities from all but the last $k$ candidates by $k$ in the corresponding proofs. Hence we turn our attention to JR. This is currently defined in the literature only for the approval setting \citep{aziz2017justified} so we first adapt it to the setting with cardinal ballots.

\begin{definition}
\label{def:jr}
    A rule $\mathcal{R}$ satisfies \emph{Justified Representation} (\emph{JR}) if for each election instance $E$, and each $(\alpha, T)$-cohesive group $S$, with $|T|=1,$
there exists a voter $i\in S$ such that $u_i(\mathcal{R}(E))>0$.
\end{definition}

Towards designing a first method with provable fairness guarantees in the examined setting it is natural to revisit the Greedy Budgeting rule from the work of \citet{Do}, which was proven to satisfy PJR under approval preferences. We begin by extending its definition to the cardinal setting.

The \emph{Greedy Budgeting rule} works as follows.  Consider an online election $E = (C, V, k, \{u_i\}_{i\in [n]})$ given as input. Each voter is initially given $\nicefrac{k}{n}$ dollars. When a
candidate $j \in C$ arrives we look if the voters in the set $A_j:=\{i\in V: u_i(j)>0\}$ have at least $1$ dollar in total. If not, we reject $j$, otherwise, we hire $j$ and ask the voters in $A_j$ to pay $1$; each voter pays the same, or all of their remaining budget. We reduce their available budget accordingly and continue with the next candidate.
If at some point the algorithm has hired $k-x$ candidates, for some $x>0$, and there are exactly $x$ candidates that remain to come, we hire all of the remaining candidates. The pseudocode of the method is presented in \Cref{alg:greedy-budgeting}.

\begin{algorithm}[t]
    \caption{Online Greedy Budgeting Rule}
    \label{alg:greedy-budgeting}
    \DontPrintSemicolon
    
    \KwIn{$E = (C, V, k, \{u_i\}_{i \in V})$}
    $m \gets |C|, n \gets |V|$
    
    $W \gets \emptyset$\tcp*{current hired committee}
    \ForEach{voter $i \in V$}{
        $b_i \gets k / n$\;
    }
    
    \ForEach{candidate $j \in C$ in arrival position $t = 1, \dots, m$}{
    \tcp*{safeguard ensuring the election of a committee of size exactly $k$} 
        $\text{remaining} \gets m - t + 1$\;
               \If{$k - |W| = \text{remaining}$}{
            add to $W$ all the remaining candidates in the stream\;
            \Return $W$\;
        }
        \tcp*{hiring decision for $j$} 
        $A_j \gets \{ i \in V : u_i(j) > 0 \}$\tcp*{supporters of $j$}
        
        \If{$\sum_{i \in A_j} b_i \geq 1$}{
                   $W \gets W \cup \{j\}$
                   \Continue\tcp*{hire $j$}
        }
reduce the budget of each voter $i\in A_j$ by $1$ in total, with each paying equally or exactly $b_i$\;}
    \Return $W$
\end{algorithm}

We show that, when applied to the cardinal setting, the rule loses some of its theoretical appeal (recall that PJR is unsatisfiable in the setting), but still manages to achieve certain proportionality guarantees, despite its simplicity. 

\begin{theorem}
\label{thm:greedy}
    Greedy Budgeting rule returns a feasible solution and satisfies JR.
\end{theorem}

\begin{proof}
Regarding feasibility, we first note that the initially distributed amount of money equals $n\cdot\nicefrac{k}{n}=k.$ Each candidate costs $1$ and is being bought only if their supporters have enough remaining funds to cover this cost. Therefore, no more than $k$ candidates can be bought by the voters. The final step of the rule ensures that it will always return a set of exactly $k$ candidates.

We will now show that the Greedy Budgeting rule satisfies JR. For contradiction, suppose that it does not. Then, there is a set of voters $S$ witnessing the violation of the axiom. So the voters of this set are at least $\nicefrac{n}{k}$ and all value positively one candidate, say $c$. Moreover, their satisfaction from the outcome is zero, but then, they have not contributed any amount of money towards including any of their approved candidates in the outcome. Therefore, at the moment of consideration of $c$, they all had their initial budget which equals $|S|\frac{k}{n}\geq \frac{n}{k}\frac{k}{n}=1$. Consequently, $c$ should have been hired at that point by the rule.
\end{proof}

While JR is a relatively weak axiom, it nonetheless guarantees a basic form of fairness. Hence, Greedy Budgeting is arguably more fair than an approach focused on maximizing total utility, which could completely overlook minority preferences and overrepresent certain groups---recall \Cref{example1}. 
Moreover, \Cref{thm:greedy} can also be strengthened to an analog of PJR, which must again refer to the number of approved candidates in the outcome and not accounting for utilities.

Below we define a stronger version of JR, motivated by the observation that \Cref{def:jr} only guarantees some positive utility from $\mathcal{R}(E)$ regardless of the value of $\alpha(c)$ when $T = \{c\}$. Instead, one could aim to adapt JR staying closer to the spirit of EJR. Such a strengthening of \Cref{def:jr} follows.

\begin{definition}
\label{def:jr-str}
    A rule $\mathcal{R}$ satisfies \emph{Strong Justified Representation} (\emph{strong-JR}) if for each election instance $E$, each $\alpha: C \mapsto \mathbb{R}$, $T=\{c\}$ for some $c\in C$, and each
$(\alpha, T)$-cohesive group of voters $S$, there exists a voter $i\in S$ such that $u_i(\mathcal{R}(E))\geq \alpha(c)$.
\end{definition}

As also supported by the results in \Cref{sec:ejr}, there is little hope when moving toward EJR. Indeed, for the stronger adaptation of JR defined in \Cref{def:jr-str}, we also get a negative result.\looseness-1

\begin{theorem}
\label{thm:jr-str}
    There is no online voting rule that satisfies strong-JR.
\end{theorem}

\begin{proof}
    Consider the simple instance where there are three candidates $\{a,b,c\}$ and two voters, $v_1,v_2$ submitting the ballots $(1,0,0)$ and $(0,1,2)$ respectively. Say that $k=2$ and that the candidate $a$ comes before $b$ who comes before $c$.
Say that $T=\{a\}$ and $S=\{v_1\}.$ For these groups it holds that $|S|\geq |T|\nicefrac{n}{k},$ and $\alpha(c)=1, \forall c\in T$. According to strong-JR, at least one voter in $S$ should receive a satisfaction of at least $1$ from the outcome.
The decision on whether to include $a$ in the winning committee or not should be done at this point and any online multi-winner election rule that rejects $a$ can't revoke this decision in case the candidates that will come do not satisfy $v_1$ to an extent of $1$. Therefore, any rule that does not have a knowledge of the preferences of $v_1$ for the candidates that are yet to come, should include $a$ in the committee at this point. Using the same argument for $T=
\{b
\}$ and $S=
\{v_2
\}$, we can show that any such rule will also include $b$ in the committee.
Hence, then, there is no place to include candidate $c$ in the solution. But, it holds that any rule that does not elect $c$ does not satisfy strong-JR. This is because if we set $T=\{c\}$ and $S=\{v_2\},$ then $S$ is $(2,T)$-cohesive. Hence, in order to satisfy strong-JR, there should be a voter $i$ in $S$ such that 
$u_i(\mathcal{R}(E))\geq 2.$ But, each voter in $S$ gets a satisfaction of at most $1$ if $c$ is not hired.
\end{proof}

\section{Proportional Selection with Risk}
\label{sec:probable}
All the negative results we unfolded regarding the satisfiability of approximate and weakened notions of EJR, 
while differing in their technical details,
are driven by the same underlying idea: towards satisfying a strong fairness guarantee in the online setting a rule must select a candidate $c$ at the moment of their arrival, provided that $c$ receives sufficient support. Under this principle, the very voters who support $c$ might later encounter a candidate they value significantly more; the online nature of the problem prevents any knowledge of future arrivals. This tension is inherent in the examined setting and any rule aiming to always act safely towards satisfying EJR must adopt such a conservative selection strategy.

In this section we propose an alternative perspective; one that accepts risk as a tool for bypassing the impossibility results of \Cref{thm:ejr1,thm:ejr2,thm:ejr3} (and \Cref{thm:jr-str}). Rather than insisting on achieving proportionality in every instance, we propose accepting occasional shortfalls.
More concretely, we will propose rules that do not commit to hiring candidates early solely because their supporters deserve some representation, 
because, by doing so, we risk exhausting the committee capacity too early, before more valuable candidates have had the chance to arrive. For the theoretical guarantees of the rules we propose hereinafter, we adopt the commonly used \emph{random arrival model} from the online voting literature \citep{gupta2021random} which assumes that elements arrive one by one in a uniformly random order.
Other models, such as the adversarial one, are natural options for follow-up work.

\subsection{Online Method of Equal Shares and Beyond}\label{sec:mes}
\label{sec:probable1}
Our approach draws inspiration from both the literature on secretary problems \citep{babaioff2007knapsack} and the most prominent voting rule for proportional representation \citep{peters2020proportionality}, combining ideas from both. 
The algorithm we propose, to be called \emph{Online Method of Equal Shares}, operates in two phases: an exploration phase of length $t$, followed by a selection phase over the remaining $m - t$ candidates. The parameter $t$, referred to as the threshold time, will be fixed later based on the analysis. Let $C_S \subseteq C$ denote the set of the first $t$ candidates of $C$, i.e., those observed during the exploration phase.\footnote{We assume without loss of generality that $t\geq k$; otherwise, we can add dummy candidates with no support.} At the end of the first phase, we construct a reference committee $C_R \subseteq C_S$ by running MES on the election induced by the first $t$ candidates. Importantly, candidates in $C_R$ will not be hired, they are only used as samples for evaluation of the quality of subsequent candidates from $C\setminus C_S$. We next proceed to the selection phase. There we maintain a running committee $C_R^r$, initialized to $C_R$. For each arriving candidate $c$, we consider the set of candidates in $C_R^r \cup \{c\},$ which are $k+1$ many. We compute the MES outcome for this candidate set. If MES returns the set $C_R^r$, then $c$ is rejected. If MES returns a new winning committee, and there is a candidate $\hat{c} \in C_R^r$ that hasn't been selected by MES (one could see it as if $c$ took the place of $\hat{c}$) we proceed as follows: We check whether $\hat{c} \in C_R$. If so, then $c$ is hired; otherwise is rejected. In other words, $c$ is hired not only if they are preferred over $\hat{c}$ but also if $\hat{c}$ was part of the initial trusted set $C_R$. Regardless of whether $c$ is hired or not, as long as $c$ was included in the outcome of MES, we update $C_R^r$ to $C_R^r \cup \{c\} \setminus \{\hat{c}\}$ and proceed with the next candidate. 
To ensure the algorithm selects exactly $k$ candidates, we include a final safeguard, the same one used in the Greedy Budgeting rule. We refer to \Cref{alg:online-mes} for a pseudocode.

\begin{algorithm}[t]
    \caption{Online Method of Equal Shares}
    \label{alg:online-mes}
    \DontPrintSemicolon
    
    \KwIn{$E=(C,V,k,\{u_i\}_{i\in V})$}
    
    $m \gets |C|, n \gets |V|$\;
    $C_S \gets$ first $\lfloor m/e \rfloor$ candidates of $C$\tcp*{exploration phase}
    $C_R \gets \textsc{MES}$ applied to the restriction of $E$ to $C_S$\tcp*{reference committee}
    $C_R^r \gets C_R, W \gets \emptyset$\tcp*{running sample committee and hired candidates}

    \ForEach{candidate $c \in C \setminus C_S$ in arrival position $i = \lfloor m/e \rfloor+1,\dots,m$}{
      \tcp*{safeguard ensuring the election of a committee of size exactly $k$} 
        $\text{remaining} \gets m - t + 1$\;
               \If{$k - |W| = \text{remaining}$}{
            add to $W$ all the remaining candidates in the stream\;
            \Return $W$\;
        }
        \tcp*{hiring decision for $c$}
        $X \gets C_R^r \cup \{c\},$
        $W' \gets \textsc{MES}$ applied to the restriction of $E$ to $X$\;
        
        Choose the $\hat{c} \in C_R^r \setminus W'$ if it exists\;
            \If{there is such a candidate $\hat{c}$ and $\hat{c} \in C_R$}{
                $W \gets W \cup \{c\}$ 
            $C_R^r \gets C_R^r \cup \{c\} \setminus \{\hat{c}\}$
        }
        
        \If{$|W| = k$}{
            \Return $W$
        }
    }
    
    \Return $W$\;
\end{algorithm}

In what follows, we present a theoretical guarantee for Online MES. Specifically, we show that it satisfies EJR with a bounded from below probability as $m$ increases. This probability depends exponentially on $k$, but also improves exponentially when aiming to satisfy the approximate notion of EJR up to $p$ candidates (see \Cref{def:ejr_approx}). 
Most importantly, this probability is independent of the total number of candidates $m$, unlike what a naive random selection would suggest. It only depends on the target committee size $k$, making the rule particularly well-suited for applications where the goal is to elect small committees---
we validate this experimentally (see for instance Experiment 2 in \Cref{sec:exp}).
We note that our theoretical guarantee holds for the single-approval case, i.e., for instances where for every voter $i$ there is at most one candidate $j$ such that $u_i(j)>0$,
yet, as we empirically show in \Cref{sec:exp}, the proposed method also succeeds in producing fair solutions for general instances.

\begin{theorem}
\label{thm:mes}    Online Method of Equal Shares returns a feasible solution and satisfies EJR up to $p$ candidates with probability $(\nicefrac{1}{e})^{k-p}$ as $m$ tends to infinity, for $t=\lfloor\nicefrac{m}{e}\rfloor$, for the single-approval case.\looseness-1
\end{theorem}

\begin{proof}
Regarding feasibility, we note that a candidate $c$ is hired only if, at the time of consideration, it manages to replace a candidate $\hat{c} \in C_R$ currently in $C_R^r$. Since $C_R$ is a feasible outcome of Offline MES, we have that $|C_R| \leq k$. Therefore, at most $k$ candidates can be removed from $C_R^r$ during the execution of the algorithm, and thus no more than $k$ candidates can be hired. The final step of the rule ensures that the output always contains exactly $k$ candidates.

We now turn to the proportionality guarantee. As a first step, we show that each candidate that would have been selected by Offline MES is hired by Online MES with probability at least $\nicefrac{1}{e}$. Fix an arbitrary $t$ and say that $c$ is a candidate that is hired under Offline MES and arrives at time $i>t$ in Online MES. Moreover say that $C_R^r$ is the running sample at the moment of appearance of $c$.
    
First, we need to show that $c$ will be included in the winning outcome of MES when the available candidates are $C_R^r \cup \{c\}$. Towards a contradiction, suppose that $c$ is not included in the winning outcome. This could happen either because the supporters of $c$ do not have enough budget to afford $c$, or because there are $k$ other candidates, say a set $C^*,$ that are $\rho$-affordable at a lower value of $\rho$. The first case cannot occur because the supporters of $c$ can only spend money towards candidate $c$, due to the single approval assumption. In the second case, all candidates in $C^*$ must have been $\rho$-affordable at the same values of $\rho$ even when considering the full set of candidates. This is because, under the single approval assumption, each voter can only contribute to the single candidate they approve. Therefore, the outcome of MES would have been $C^*$ or would consist of candidates that are $\rho$-affordable for even lower values of $\rho$ than the candidates in $C^*$, contradicting the fact that $c$ is a MES-winning candidate.
    
Under online-MES, $c$ will be selected only if the candidate $\hat{c}$ of $C_R^r$ that is being kicked out of the MES-winning committee by $c$ at the moment of consideration of $c$ was in the reference set $C_R$ found during the sampling phase. Since we have assumed uniformly random permutation of candidates, this happens with probability $\frac{t}{i-1},$ since $i-1$ spots of appearance of $\hat{c}$ exist before the appearance of $c$ and only $t$ out of them are leading to the selection of $c$. Moreover, each step $i$ is equally likely for $c$ to be spotted then (with probability of $\nicefrac{1}{m}$), so the probability of selecting $c$ is 
\setlength{\abovedisplayskip}{2pt}
\setlength{\belowdisplayskip}{2pt}
\begin{equation*}
\sum_{t+1}^m \frac{1}{m}\frac{t}{i-1}=\frac{t}{m}\sum_{t+1}^m \frac{1}{i-1}>\frac{t}{m} \ln{\frac{m}{t}}.
\end{equation*}
The function $\frac{t}{m} \ln{\frac{m}{t}}$ is being maximized for $t=\nicefrac{m}{e}$. Hence, for $t=\lfloor \nicefrac{m}{e} \rfloor$, with a probability of at least $\nicefrac{1}{e},$ the candidate $c$ is hired by Online MES. 
If all the $k$ candidates of the MES-winning committee are in the winning committee of Online MES, then its outcome satisfies EJR. This is a consequence of the fact that MES satisfies EJR for multi-winner elections 
\citep{peters2021proportional} and happens with a probability of $(\nicefrac{1}{e})^k.$ Moreover, to satisfy EJR up to $p$ candidates, it suffices to elect only $k-p$ candidates from the outcome of MES, and this happens with a probability of $(\nicefrac{1}{e})^{k-p}.$ 
\end{proof}

The proof of \Cref{thm:mes} directly extends to further guarantees satisfied by Offline MES, including EJR+ 
\citep{brill2023robust} and other notions of representation (see the book by \citet{lackner2023chap4} for more). EJR+ is one of the strongest satisfiable proportionality axioms known in the literature, and strictly stronger than EJR. It is violated by a candidate not in the winning set if they are approved by $\ell\frac{n}{k}$ voters, each of whom approves less than $\ell$ members of the outcome.\looseness-1

Probabilistically guaranteeing proportionality axioms (via deterministic rules) is a conceptually appealing and novel approach. Online MES is the first rule providing such bounded guarantees; whether analyses improving the one in \Cref{thm:mes} are possible remains open. It is also the first classic offline voting rule adapted to the online setting.  
Yet, the main appeal of the method is explained in the next paragraph.

\paragraph{A Framework for Bringing Voting Rules Online \& the Case of MES with Bounded Overspending}
The principle behind Online Method of Equal Shares has a broader applicability. We highlight that its mechanics are not specific to MES, and can be adapted to work with any voting rule. 
In that sense, our method can be seen as an instance of \textit{a general framework that converts any rule for offline elections into one that works in the online setting}.
MES is the most widely known and used voting method for achieving proportional outcomes.
However, it is not without drawbacks. As discussed in the very recent work by \citet{papasotiropoulos2024method}, there are natural cases where MES behaves in counterintuitive ways. One such case arises precisely due to the use of cardinal ballots---which is exactly the focus of our work.
A simple but insightful example, known as the “weakness of tail utilities,” captures this issue well. 
To address limitations of the Method of Equal Shares,
\citet{papasotiropoulos2024method} introduced the Method of Equal Shares with Bounded Overspending (BOS).
BOS is considered better suited for practical applications than MES, with the council of Pruszków, Poland, already adopting it for an election featuring the largest budget to date allocated using a voting rule with theoretical fairness guarantees.
The approach we proposed in \Cref{sec:mes} is fully compatible with BOS as well and it can be applied using BOS as the underlying rule instead of MES. 
In particular, it suffices to execute \Cref{alg:online-mes} using BOS instead of MES in the step that creates the reference committee as well as for the computation of $W'$.
Therefore, in our experiments, (\Cref{sec:exp}) we also consider \emph{Online BOS} and we show that it exhibits exceptionally good behavior.

\subsection{Online Nash Product Rule}
\label{sec:nash}
Given an offline election, a natural first objective is to maximize the total voter utility, hence elect the candidates in $\arg \max_{W\in \mathcal{W}}\sum_{i\in V}u_i(W).$
As observed by \citet{Do}, the approval-based variant of this problem essentially matches the multi-secretary problem studied by \citet{kleinberg2005multiple}. The online algorithm proposed there applies directly to the setting with cardinal ballots as well, yielding a constant competitive ratio against the offline optimum.
However, this objective entirely overlooks fairness among voters. It values all gains equally, without considering who benefits. 
To address this, a more fairness-aware objective is to \emph{maximize the product of utilities}. This approach promotes balanced distribution of voter satisfaction and,
here, fairness is understood as ensuring equitable utility among voters, rather than proportional representation as captured by EJR and related axioms discussed in \Cref{sec:certainty,sec:mes}. This 
objective is known as the \emph{Nash Product rule}. Formally, it is defined as $\arg\max_{W \in \mathcal{W}} \prod_{i \in V} u_i(W)$ \citep{nash1950bargaining}, or, equivalently in terms of maximizers,  
$\arg\max_{W \in \mathcal{W}} \sum_{i \in V} \log(1 + u_i(W)),$ 
where the addition of $1$ ensures that all terms are well-defined, even when some voters have a utility of zero \citep{fluschnik2019fair}. We call \emph{Nash Welfare} the quantity $\Pi(W):=\sum_{i \in V} \log(1 + u_i(W))$.
This and related formulations of the Nash Product rule are well-studied in computational social choice
\citep{caragiannis2019unreasonable,freeman2017fair,conitzer2017fair,fain2018fair,ebadian2024optimized}. 
Considering a utility-product formulation in the online setting is an interesting direction for future work which would likely require a different approach than the one we propose here.

We introduce this rule to a new application domain: the fair selection of candidates in an online secretary hiring framework.
From the original axiomatic characterization \citep{nash1950bargaining}, we can derive that it is the only one satisfying several important axioms—each natural and desirable in our setting, even when fairness is not the main objective. According to this characterization, a more equal distribution of satisfaction among voters is better, increasing a voter's satisfaction without decreasing that of others is also an improvement, and the optimal outcome remains unchanged under scaling of valuations;
we refer to the work of \citet{kilgourhandbook} for details and to the papers by \citet{ramezani2009nash} and \citet{lu2024approval} for further properties.\looseness-1

The Nash Product rule can be seen as an analog of PAV \citep{fluschnik2019fair}, a rule with strong guarantees of proportionality in the approval setting \citep{lackner2023chap4}
which does not extend directly to settings with cardinal utilities.
 On the negative side, likewise PAV, computing a committee that maximizes the Nash Welfare is NP-hard in the offline setting \citep{ramezani2009nash}. 
 The committee returned by our method, which runs in polynomial time, achieves a Nash Welfare within a multiplicative factor of roughly $\nicefrac{1}{11}$ compared to the offline optimal.

The algorithm we propose, to be called \emph{Online Nash Rule} works as follows.
We divide the input stream into $k$ consecutive and equally-sized\footnote{We assume without loss of generality that $m$ is a multiple of $k$; otherwise, we can add dummy candidates with no support arriving in random order.} segments, with the goal of selecting one candidate per segment. Intuitively, in each segment we aim to select the best available candidate based on the contributions to the overall Nash Welfare.
Let $T_0 = \emptyset$ and, for $i\in \{1,2,\dots,k\}$ let $T_i$ denote the set of candidates selected by the time a decision has been made for the last candidate of the $i$-th segment. For each segment $i \in \{1, \dots, k\}$, we define the marginal gain function $g_i(c) = \Pi(T_{i-1} \cup \{c\})$, which measures the value of adding the candidate $c$ to the set of already selected candidates with respect to the function $\Pi$ that denotes the Nash Welfare. Our objective in segment $i$ is to select the candidate maximizing $g_i$.
To do so, we apply the classical algorithm for the secretary problem within each segment \citep{dynkin1963optimum}.
Specifically, for the $i$-th segment, we observe the first $\lfloor\nicefrac{1}{e}\rfloor$-fraction of candidates of the segment, say $R_i$, without making any selection. Let $\hat{c}_i:=\arg\max_{c\in R_i}g_i(c)$. Then, we include in $T_i$ the first candidate that comes after this observation phase and has $g_i(c) \geq g_i(\hat{c}_i)$. If no such candidate arrives, we select the last candidate of the segment.
The final selected committee is $T_k$. \Cref{alg:online-nash} presents a pseudocode of the method.

\begin{algorithm}[t]
    \caption{Online Nash Rule}
    \label{alg:online-nash}
    \DontPrintSemicolon
    
    \KwIn{$E=(C,V,k,\{u_i\}_{i\in V})$}
    
    $m \gets |C|$, $n \gets |V|$ \;
    Partition $C$ into $k$ consecutive equally-sized segments $S_1,\dots,S_k$ according to the arrival order\;
    $T_0 \gets \emptyset$
    
    \For{$i \gets 1$ \KwTo $k$}{
        $R_i \gets$ first $\lfloor |S_i|/e \rfloor$ candidates of $S_i$\tcp*{observation phase of the $i$-th segment}
        
        \ForEach{$c \in R_i$}{
            $g_i(c) \gets \Pi(T_{i-1} \cup \{c\})$\;
        }
        
        $\hat{c}_i \gets \arg\max_{c \in R_i} g_i(c)$\;
        
        \ForEach{remaining candidate $c$ in $S_i \setminus R_i$ in order}{
            $g_i(c) \gets \Pi(T_{i-1} \cup \{c\})$\;
            \If{$g_i(c) \geq g_i(\hat{c}_i)$}{
                $c^\star \gets c$\tcp*{hire $c$ as long as it is better than the sample's best}
                \textbf{break}\;
            }
            \If{$c$ is the last candidate of $S_i$}{
                $c^\star \gets c$\tcp*{fallback: hire the last candidate of the segment}
            }
        }
        
        $T_i \gets T_{i-1} \cup \{c^\star\}$\tcp*{election of the $i$-th member of the committee}
    }
    
    \Return $T_k$\;
\end{algorithm}

\begin{theorem}
\label{thm:nash}
    Online Nash Rule returns a feasible solution and the expected value of its outcome is at least $\frac{1-\nicefrac{1}{e}}{7}$ of the offline maximum Nash Welfare. 
\end{theorem}
\begin{proof}
\citet{bateni2013submodular} proposed an online algorithm that, given a monotone nonnegative submodular function $f$, returns a solution whose expected value under $f$ is within a multiplicative factor of $\frac{1 - \nicefrac{1}{e}}{7}$ of the optimum. The Online Nash Rule is simply an application of the rule by \citet{bateni2013submodular} with the Nash Welfare function $\Pi$ as the maximizing function $f$. Therefore, to prove the claimed guarantee it suffices to show that the Nash Welfare function is a monotone nonnegative submodular one. The first two properties are trivial by definition. It remains to prove submodularity. 
    
    We want to show that the function  
$
\Pi(W) = \sum_{i \in V} \log(1 + u_i(W))
$  
is submodular.  
Recall that a function $f: 2^C \mapsto \mathbb{R}_{\geq 0}$ is submodular if for all $A \subseteq B \subseteq C$ and every $c \in C \setminus B$, it holds
$
f(A \cup \{c\}) - f(A) \geq f(B \cup \{c\}) - f(B).
$
 We want to show that this holds for $\Pi$. Fix any sets of candidates $A,B$ such that $A\subseteq B$ and a candidate $c\notin B$. Then, for both $X=A$ and $X=B$ it holds that
\begin{align*}
\Pi(X\cup\{c\}) - \Pi(X) &=
\sum_{i \in V} \left[ \log(1 + u_i(X \cup \{c\})) - \log(1 + u_i(X)) \right] =\\
\sum_{i \in V} \left[ \log(1 + u_i(X)+u_i(c)) - \log(1 + u_i(X)) \right] &=
\sum_{i \in V} \log \frac{1 + u_i(X)+u_i(c)}{1 + u_i(X)}= \sum_{i \in V} \log(1+\frac{u_i(c)}{1 + u_i(X)}).
\end{align*}

Fix an arbitrary voter $i\in V$. Since $A\subseteq B$ it holds that $u_i(A)\leq u_i(B)$. Therefore, $\frac{1}{u_i(A)}\geq \frac{1}{u_i(B)},$ and the hence we get
$\Pi(A\cup\{c\}) - \Pi(A) \geq \Pi(B\cup\{c\}) - \Pi(B)$. The feasibility of the outcome follows directly from the mechanics of the rule. 
\end{proof}

\section{Empirical Evaluation}
\label{sec:exp}
Finally, we experimentally study Greedy Budgeting, Online MES,\footnote{Offline MES is typically paired with a completion method to ensure the selection of exactly $k$ candidates \citep{faliszewski2023participatory,brill2024completing}. Several such methods have been proposed in the literature. Since Online MES relies on Offline MES as a subroutine, we had to choose a specific variant for our experiments---though this choice does not affect the theoretical guarantees of our paper. We tested both the Utilitarian and Add1U completion methods on small instances and found their outcomes to be largely similar. Given the already high running time of Online MES without completion, we focused on the Utilitarian variant for efficiency.
} Online BOS, and Online Nash rules on both real-world (Experiments 1 and 2) and synthetic datasets (Experiments 3 and 4), designing a set of simulations with distinct goals and evaluation metrics.
The key takeaways are that (i) all the proposed rules perform well in terms of fairness across a range of setups and metrics, (ii) when comparing them directly, Online BOS stands out for its consistency: it is the only rule among those we examined that is consistently either the best or close to it, (iii) Online BOS outperforms Online MES across all experiments and setups.
The code for replicating the experiments can be found at
\url{https://github.com/PishbinZein/FairnessInTheMultiSecretaryProblem_AAAI2026}.

\subsection[Pabulib Instances]{Experiment 1: Pabulib Instances}
We computed the outcomes of the proposed rules across 
all the $773$ participatory budgeting instances that were available in the \pabulib\ library \citep{faliszewski2023participatory} as of 22-08-2025. To allow the experiment to run in reasonable time, we only focused on instances of manageable size and constructed the dataset as follows. We used all approval-based instances (excluding those with $m \leq 2$) where the number of voters was not excessively large ($n \leq 50,\!000$). Among those, we selected instances where either both $n$ and $m$ are small ($n < 10,\!000$ and $m < 100$), or where $n$ is large ($10,\!000 < n < 50,\!000$) but $m$ is small ($m < 80$), or where $m$ is large ($m > 100$), but $n$ is small ($n < 5,\!000$). We then partitioned these instances into three groups based on the number of candidates $m$: instances with $m < 10$ were considered small (202 instances), those with $m \geq 30$ were considered large (396 instances), and the rest were classified as medium (175 instances).
For each instance, we consider $5$ different values of committee size, chosen as functions of $m$ to allow consistent comparison across datasets. Additionally, for each setting, we run $5$ iterations with candidate arrival orders drawn uniformly at random, in line with our theoretical analyses. 
To assess proportionality, we use EJR+ as its satisfaction (and degree of violation) can be verified efficiently.
In \Cref{tab:exp1app1,tab:exp1app2} we report the (average) share of voters who are not adequately represented and the number of candidates witnessing the violation, respectively, for various values of $k$, 
with instances broken down by the number of candidates.

\begin{table}[t]
\centering
\small
\resizebox{0.715 \linewidth}{!}{%
\begin{tabular}{lcccccc}
\toprule
 & $k = m/20$ & $k = m/15$ & $k = m/10$ & $k = m/7$ & $k = m/4$  \\
\midrule
\multicolumn{7}{c}{\textbf{Greedy Budgeting}} \\
\midrule
Small instances  & 0.0000 & 0.0000 & 0.0000 & 0.0137 & 0.0117  \\
Medium instances & 0.0242 & 0.0200 & 0.0188 & 0.0181 & 0.0156  \\
Large instances  & 0.0176 & 0.0171 & 0.0164 & 0.0155 & 0.0135  \\
All instances    & 0.0206 & 0.0188 & 0.0181 & 0.0167 & 0.0141  \\
\midrule
\multicolumn{7}{c}{\textbf{Online MES}} \\
\midrule
Small instances  & 0.0000 & 0.0000 & 0.0000 & 0.0137 & 0.0116  \\
Medium instances & 0.0236 & 0.0196 & 0.0188 & 0.0181 & 0.0157 \\
Large instances  & 0.0177 & 0.0172 & 0.0165 & 0.0157 & 0.0135  \\
All instances    & 0.0203 & 0.0186 & 0.0181 & 0.0168 & 0.0142 \\
\midrule
\multicolumn{7}{c}{\textbf{Online BOS}} \\
\midrule
Small instances  & 0.0000 & 0.0000 & 0.0000 & 0.0137 & 0.0116  \\
Medium instances & 0.0236 & 0.0196 & 0.0185 & 0.0180 & 0.0155  \\
Large instances  & 0.0176 & 0.0171 & 0.0163 & 0.0155 & 0.0133  \\
All instances    & 0.0203 & 0.0186 & 0.0179 & 0.0167 & 0.0140 \\
\midrule
\multicolumn{7}{c}{\textbf{Online Nash Rule}} \\
\midrule
Small instances  & 0.0000 & 0.0000 & 0.0000 & 0.0138 & 0.0117  \\
Medium instances & 0.0226 & 0.0190 & 0.0194 & 0.0187 & 0.0164  \\
Large instances  & 0.0181 & 0.0177 & 0.0170 & 0.0161 & 0.0140  \\
All instances    & 0.0201 & 0.0185 & 0.0186 & 0.0172 & 0.0147  \\
\bottomrule
\end{tabular}%
}
\caption{Average EJR+ Violations by Rule and Committee Size (Experiment 1). Values represent percentages of voters witnessing a villation, where 0 corresponds to 0\% and 1 to 100\%}
\label{tab:exp1app1}
\end{table}

\begin{table}[t]
\centering
\small
\resizebox{0.715 \linewidth}{!}{%
\begin{tabular}{lcccccc}
\toprule
 & $k = m/20$ & $k = m/15$ & $k = m/10$ & $k = m/7$ & $k = m/4$  \\
\midrule
\multicolumn{7}{c}{\textbf{Greedy Budgeting}} \\
\midrule
Small instances  & 0.0000 & 0.0000 & 0.0000 & 0.3656 & 0.5192  \\
Medium instances & 0.2257 & 0.2423 & 0.3388 & 0.4915 & 0.9038  \\
Large instances  & 0.3380 & 0.4883 & 0.7399 & 1.0469 & 1.7986  \\
All instances    & 0.2877 & 0.3436 & 0.4617 & 0.6138 & 1.0142  \\
\midrule
\multicolumn{7}{c}{\textbf{Online MES}} \\
\midrule
Small instances  & 0.0000 & 0.0000 & 0.0000 & 0.3566 & 0.5115  \\
Medium instances & 0.2228 & 0.2375 & 0.3305 & 0.4800 & 0.9079  \\
Large instances  & 0.3313 & 0.4800 & 0.7377 & 1.0623 & 1.8746  \\
All instances    & 0.2827 & 0.3374 & 0.4553 & 0.6097 & 1.0319 \\
\midrule
\multicolumn{7}{c}{\textbf{Online BOS}} \\
\midrule
Small instances  & 0.0000 & 0.0000 & 0.0000 & 0.3566 & 0.5088  \\
Medium instances & 0.2228 & 0.2375 & 0.3298 & 0.4781 & 0.8939 \\
Large instances  & 0.3306 & 0.4775 & 0.7254 & 1.0241 & 1.7721  \\
All instances    & 0.2823 & 0.3363 & 0.4511 & 0.5988 & 1.0003  \\
\midrule
\multicolumn{7}{c}{\textbf{Online Nash Rule}} \\
\midrule
Small instances  & 0.0000 & 0.0000 & 0.0000 & 0.3709 & 0.5408  \\
Medium instances & 0.2296 & 0.2444 & 0.3424 & 0.5180 & 1.0003  \\
Large instances  & 0.3478 & 0.5104 & 0.7995 & 1.1666 & 2.0922 \\
All instances    & 0.2948 & 0.3540 & 0.4825 & 0.6609 & 1.1375 \\
\bottomrule
\end{tabular}%
}
\caption{Average Shortfall of EJR+ Violations by Rule and Committee Size (Experiment 1).}
\label{tab:exp1app2}
\end{table}

Overall, the experiment showcases that all of our rules perform very well with respect to EJR+ violations: the violating voters percentage is only marginally different across rules and there are around 1\% to 2.5\% of voters who are underrepresented on average; the average shortfall indicates that all rules either satisfy EJR+ up to one candidate (see \Cref{def:ejr_approx}) or come very close to doing so. As $k$ increases, the percentage of violating voters decreases, while the shortfall increases.
The objective of this experiment is not to draw comparisons between the rules; the subsequent experiments in this section are devoted to comparative analyses.
Finally, in \Cref{tab:exp1runtime}, we report the average running time\footnote{The experimental part of our work was implemented in Python and run on an ASUS VivoBook laptop with an Intel i7‑11370H processor, 16GB of RAM, and a 512GB PCIe SSD.} per instance for each rule, showcasing a noticeable difference among them---we briefly return to it in \Cref{sec:conclusion}.

\begin{table}[h!]
\centering
\resizebox{0.75 \linewidth}{!}{%
\begin{tabular}{lc}
\toprule
\textbf{Rule} & \textbf{Time (seconds)} \\
\midrule
Greedy Budgeting Rule  & 0.0181 \\
Online Method of Equal Shares          & 1.1999 \\
Online Method of Equal Shares with Bounded Overspending          & 6.7802 \\
Online Nash Rule    & 0.3237 \\
\bottomrule
\end{tabular}%
}
\caption{Average running time per instance from \pabulib, for each rule (Experiment 1).}
\label{tab:exp1runtime}
\end{table}

\subsection[Sushi and MovieLens Datasets]{Experiment 2: Sushi and MovieLens Datasets}
A limitation of the datasets used in Experiment 1 is that they come from processes where candidates are project proposals with associated costs---which we had to ignore since our rules are not designed for cost-based settings. 
Additionally, since Experiment 1 focused on the EJR+ metric which is only defined for approval preferences, the datasets used involve only binary preferences. Our second experiment utilizes two datasets which avoid these limitations and are standard benchmarks in computational social choice and beyond. 

The Sushi dataset \citep{kamishima2003nantonac, kamishima2010survey} contains 5,000 individual rankings of 100 items grouped into 5 preference tiers with some items left unranked, which we interpret as scores from 0 to 5. We evaluate 9 different values of $k$; specifically $k\in \{2,5,8,15,20,25,35,40,45\}$.
For each, we run 20 iterations of each algorithm over random arrival orders of candidates. 

The MovieLens dataset we used \citep{harper2015movielens} contains 100,000 movie ratings where 1000 users are giving scores from $0$ to $5$ on 1700 movies. The same values of $k$ were considered for this dataset. We run 50 iterations of each algorithm, each with a random arrival order.
The two datasets differ significantly in the number of candidates, a deliberate choice to provide insights into rules' behavior under different candidates' set sizes.

The goal here is to analyze the relative performance of the rules under varying committee sizes.  
We evaluate performance using four metrics: average voter satisfaction, percentage of voters with a satisfaction of $0$ (exclusion ratio), average satisfaction of the least satisfied quartile of voters, and the Gini coefficient capturing inequality in satisfaction.
In \Cref{tab:exp21} we report, for each metric, the percentage that a rule ranks first, or among the two best or last among the examined rules, for different committee sizes, for the Sushi dataset. \Cref{tab:exp22} report the results for the MovieLens dataset.

For the Sushi dataset, Greedy Budgeting is most often the top-performing rule across all the examined metrics, but there are also multiple instances where it fails to rank among the
top two. Online BOS, on the other hand, appears more consistently reliable: for three metrics it ranks among the two best significantly more frequently than Greedy, and even for the 25th percentile metric---where it is slightly behind---it still does so in the vast majority of instances. The MovieLens dataset is considerably more favorable to Online BOS, which ranks first in far more instances than any other rule.
Online MES also sees a notable increase in the percentage of instances where it ranks among the top two. These shifts can be explained by the increase in the number of available candidates (from 100 in the Sushi dataset to 1700 in MovieLens), which, with $k$ fixed, inherently favors Online MES and BOS by design.

\begin{table}[t]
\centering
\resizebox{\linewidth}{!}{%
\begin{tabular}{@{}c|c|c|c|c@{}}
\toprule
\textbf{Best | Top-2 | Worst} & \textbf{Avg Sat} & \textbf{Excl Ratio} & \textbf{25th Perc} & \textbf{Gini} \\
\midrule
\textbf{Greedy Budgeting} & 47.78\% | 72.22\% | 3.33\% & 56.11\% | 74.44\% | 3.33\% & 64.44\% | 91.67\% | 0.00\% & 48.89\% | 73.89\% | 3.33\% \\
\textbf{Online MES} & 8.33\% | 32.78\% | 3.89\% & 8.89\% | 33.33\% | 5.00\% & 4.44\% | 27.78\% | 2.78\% & 8.89\% | 30.56\% | 5.00\% \\
\textbf{Online BOS} & 43.33\% | 93.89\% | 1.11\% & 34.44\% | 90.56\% | 1.11\% & 31.11\% | 80.56\% | 0.00\% & 41.67\% | 94.44\% | 1.11\% \\
\textbf{Online Nash Rule} & 0.56\% | 1.11\% | 91.67\% & 0.56\% | 1.67\% | 90.56\% & 0.00\% | 0.00\% | 97.22\% & 0.56\% | 1.11\% | 90.56\% \\
\bottomrule
\end{tabular}%
}
\caption{Percentage where each rule performed best, among the two best and worst, according to different metrics, for the Sushi dataset (Experiment 2).}
\label{tab:exp21}
\end{table}

\begin{table}[t]
\centering
\resizebox{\linewidth}{!}{%
\begin{tabular}{@{}c|c|c|c|c@{}}
\toprule
\textbf{Best | Top-2 | Worst} & \textbf{Avg Sat} & \textbf{Excl Ratio} & \textbf{25th Perc} & \textbf{Gini} \\
\midrule
\textbf{Greedy Budgeting} & 4.00\% | 22.89\% | 0.00\% & 11.78\% | 54.67\% | 0.00\% & 12.89\% | 56.67\% | 0.00\% & 14.67\% | 58.00\% | 0.00\% \\
\textbf{Online MES} & 8.67\% | 80.00\% | 0.67\% & 7.78\% | 48.22\% | 0.67\% & 3.56\% | 53.56\% | 0.00\% & 6.67\% | 44.67\% | 0.67\% \\
\textbf{Online BOS} & 87.33\% | 96.67\% | 0.89\% & 80.44\% | 96.67\% | 0.89\% & 83.56\% | 89.78\% | 0.00\% & 78.67\% | 96.89\% | 0.89\% \\
\textbf{Online Nash Rule} & 0.00\% | 0.44\% | 98.44\% & 0.00\% | 0.44\% | 98.44\% & 0.00\% | 0.00\% | 100.00\% & 0.00\% | 0.44\% | 98.44\% \\
\bottomrule
\end{tabular}%
}
\caption{Percentage where each rule performed best, among the two best and worst, according to different metrics, for the MovieLens dataset (Experiment 2).}
\label{tab:exp22}
\end{table}

\subsection[Sampled Datasets]{Experiment 3: Sampled Datasets}
The two most commonly used models for generating synthetic data in computational social choice are Impartial Culture (IC) and Mallows model \citep{boehmer2024guide}.
We evaluate our rules using these and we also explore a recently proposed culture for synthetic preference generation that is well suited for experiments involving varying numbers of candidates: the Normalized Mallows model \citep{pmlr-v202-boehmer23b}. 

For IC, we select 5 fixed values of a parameter $p\in \{0.2, 0.4, 0.6, 0.8, 1.0\}$ and we say that each voter approves a number of candidates drawn from a binomial distribution centered around the fraction $p$ of $m$. Having fixed their number, the actual approved candidates are chosen at random. Then, positive utilities capped at 200 are assigned following a normal distribution with mean value 150 and standard deviation 140. We examine the values of $n$ from the set $\{5, 10, 20, 30, 50, 70, 100, 150\}$ and we also fix $32$ pairs of $(m,k)$ values selected so as for each $m$ to have both small, medium and large fraction of it as $k$. For each combination, we run $10$ iterations with different candidate arrival orders.

For Mallows model, we first generate ranked ballots in the usual way \citep{mallows1957non,boehmer2022expected,boehmer2021putting}, and then we assign utilities evenly spaced between 0 and 200 based on the ranking. We also add some utility noise across voters for achieving larger variance of scores in ballots. Again, we fix $n$, $m$, and $k$ as in the IC model and evaluate our rules on profiles with $5$ different dispersion parameters $\phi \in \{0.2, 0.4, 0.6, 0.8, 1\}$ and, again, $10$ different candidate arrival orders. 
Regarding the parameterization of the experiment for Normalized Mallows, we set it up in the exact same way as for the classic Mallows model and we follow the implementation from Python PrefSampling library \citep{prefsampling}.

Even in the well-studied offline setting of multi-winner elections, there is no single widely accepted metric for measuring proportionality. However, the Method of Equal Shares is highly regarded for its proportionality guarantees. We assess how much our rules deviate from the Offline MES outcomes based on the four statistics also used in Experiment 2. Specifically, we examined the average satisfaction of all voters and of the least satisfied quarter of voters as well as the Gini coefficient and exclusion ratio as measures of fairness. 
Naturally, given their online nature, our rules fall short compared to Offline MES, but the key question is by how much, and which rule performs best in minimizing this gap.

We first pay particular attention to the exclusion ratio metric for the Impartial Culture model against the approval probability $p$ (\Cref{fig:3a/5}, bottom-right). When voters approve many candidates, each with varying scores, all rules exhibit an exceptionally low exclusion ratio on average, and as the approval probability decreases, differences between the rules become apparent. Across all values of $p$, the exclusion ratio of all four rules remain within an additive gap of at most $2.5$ candidates from that of Offline MES. For low approval probabilities, BOS and Greedy perform nearly as well, being the top performers, with Online MES staying slightly behind. For $p=0.6$, Greedy becomes the worst rule among the examined, while BOS achieves again a guarantee particularly close to the one provided by Offline MES. 

In the same sampling model (IC), the Gini coefficient (\Cref{fig:3a/5}, top-left) of BOS is consistently the smallest of our proposed rules and the gap between each rule and Offline MES narrows as $n$ increases. 
In \Cref{fig:3b/5,fig:3c/5} (top-left), we also report the Gini coefficient of each rule’s outcome for the Mallows cultures, the values of which remain consistently particularly low for all rules, with their ranking following the same order as in the case of IC.

For the two satisfaction-based metrics (average and worst quartile), the conclusions are broadly similar: Online BOS consistently ranks highest, followed by Online MES with a clear margin over the remaining rules, and, as $n$ increases, the performance of all rules approaches that of Offline MES. These patterns hold across all three cultures (see also top-right and bottom-left of \Cref{fig:3a/5,fig:3b/5,fig:3c/5})
and can be observed also from the plots showing the 25th percentile utility in relation to the dispersion parameter of the Mallows models (bottom-right of \Cref{fig:3b/5,fig:3c/5}).

\begin{figure}[h!]
    \centering
    \begin{subfigure}{0.496\textwidth}
        \centering
        \includegraphics[width=\textwidth]{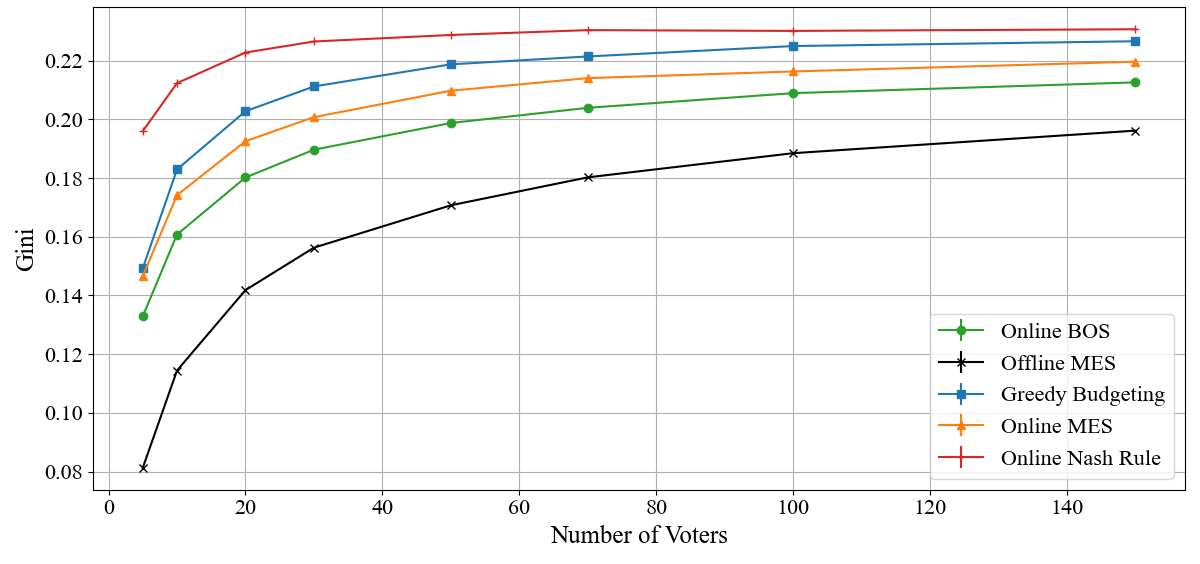}
    \end{subfigure}
    \hfill
    \begin{subfigure}{0.496\textwidth}
        \centering
        \includegraphics[width=\textwidth]{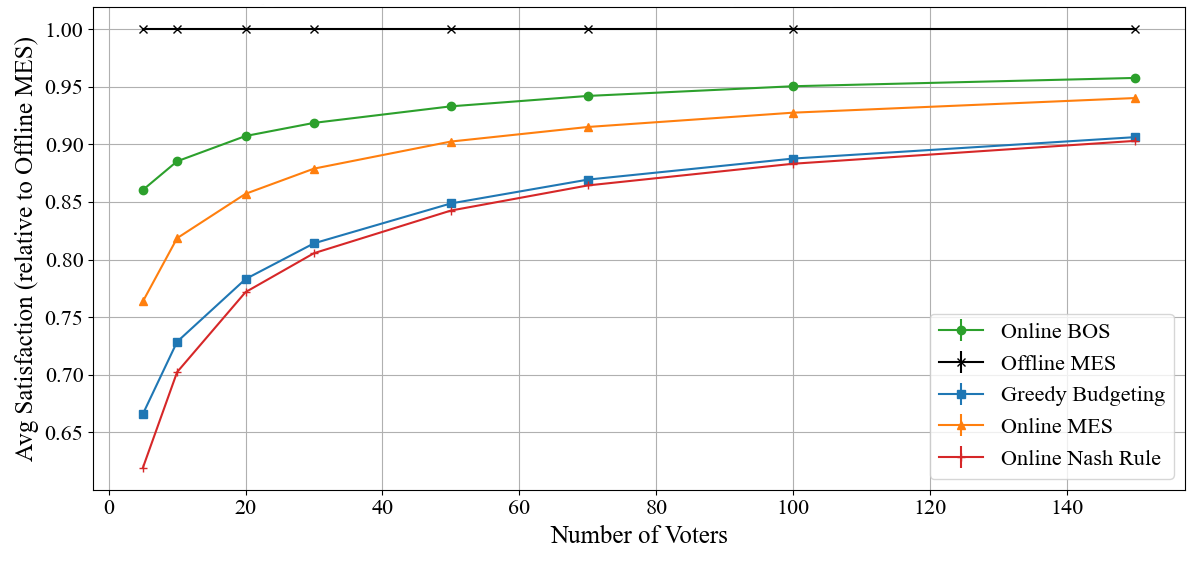}
    \end{subfigure}
    
    \vspace{0.3cm}
    
    \begin{subfigure}{0.496\textwidth}
        \centering
        \includegraphics[width=\textwidth]{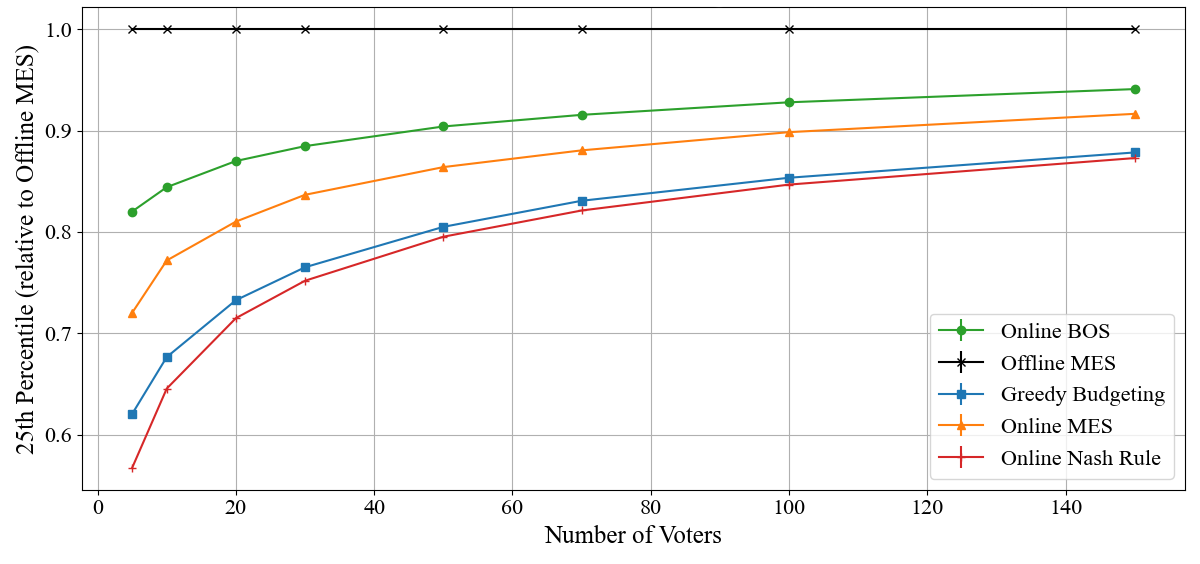}
    \end{subfigure}
    \hfill
    \begin{subfigure}{0.496\textwidth}
        \centering
        \includegraphics[width=\textwidth]{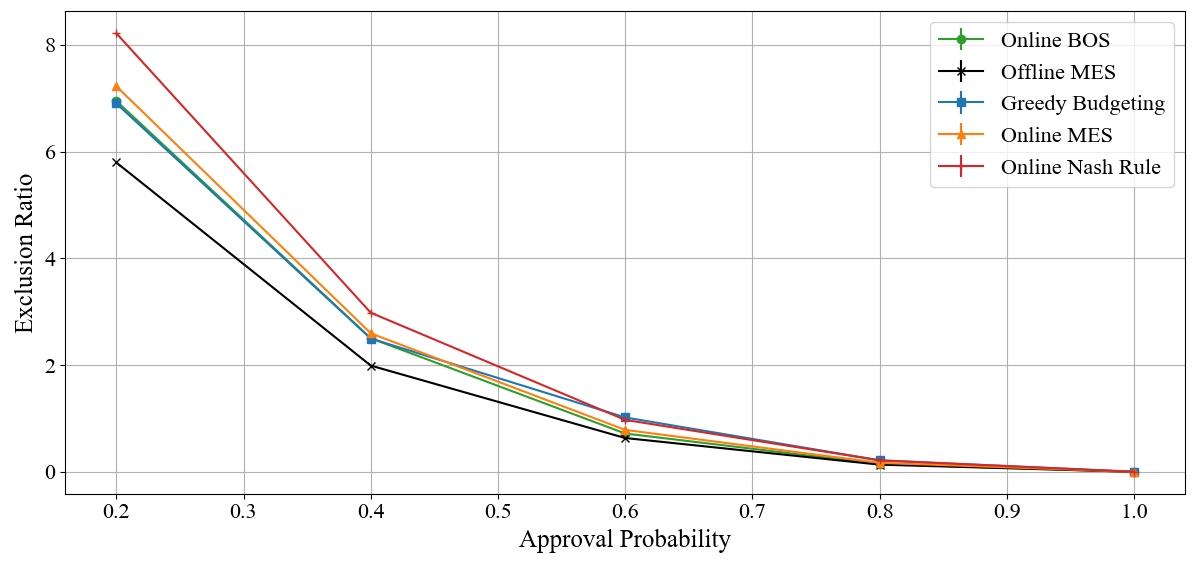}
    \end{subfigure}
    \caption{Gini coefficient against the number of voters (top-left), average satisfaction (top-right) and 25th percentile (bottom-left) relative to Offline MES against the number of voters, and Exclusion ratio against the approval probability $p$ (bottom-right) in the IC model (Experiment 3).}
    \label{fig:3a/5}
\end{figure}
\begin{figure}[h!]
    \centering
    \begin{subfigure}{0.496\textwidth}
        \centering
        \includegraphics[width=\textwidth]{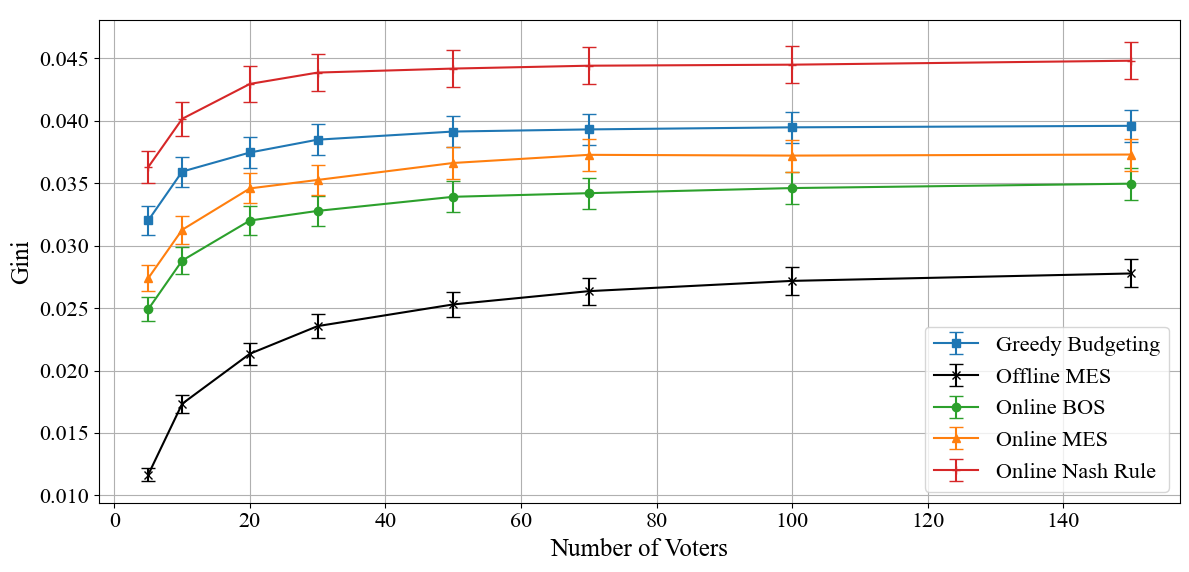}
    \end{subfigure}
    \hfill
    \begin{subfigure}{0.496\textwidth}
        \centering
        \includegraphics[width=\textwidth]{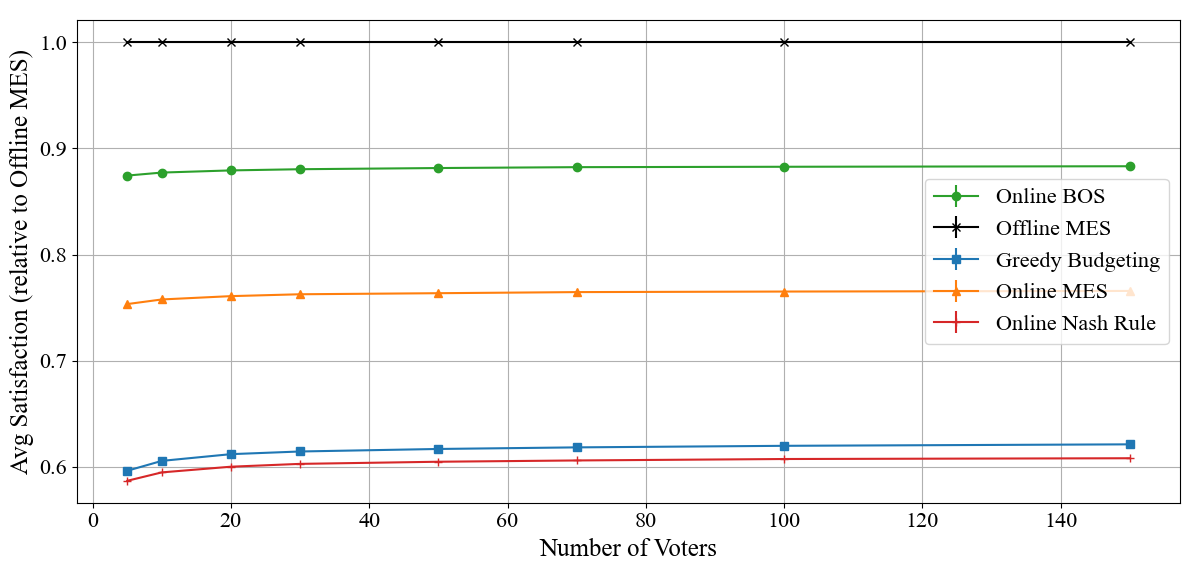}
    \end{subfigure}
    
    \vspace{0.5cm}
    
    \begin{subfigure}{0.496\textwidth}
        \centering
        \includegraphics[width=\textwidth]{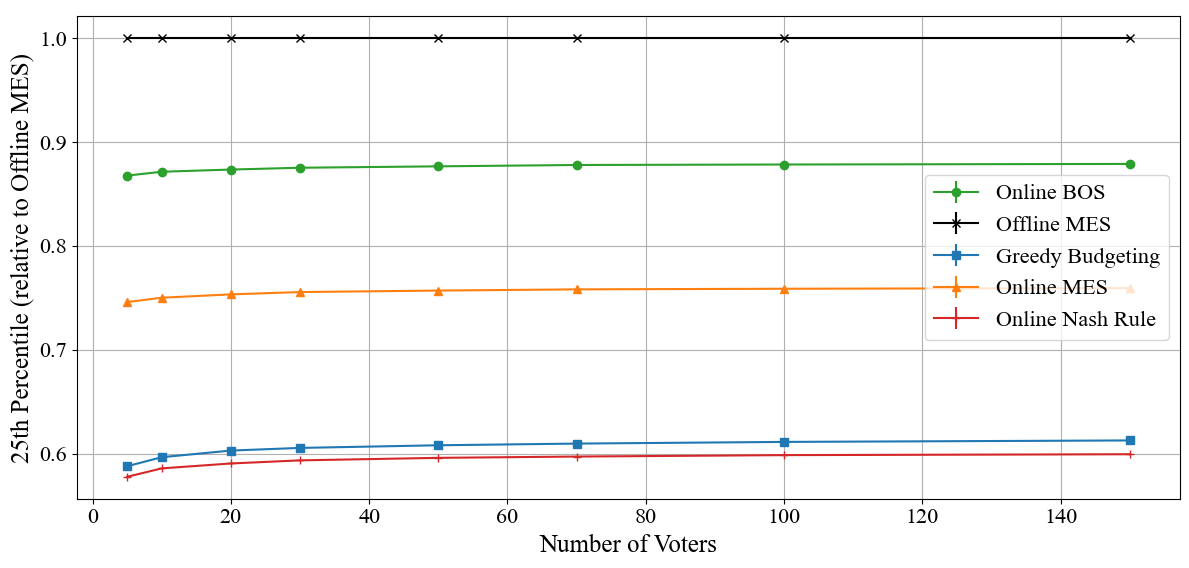}
    \end{subfigure}
    \hfill
    \begin{subfigure}{0.496\textwidth}
        \centering
        \includegraphics[width=\textwidth]{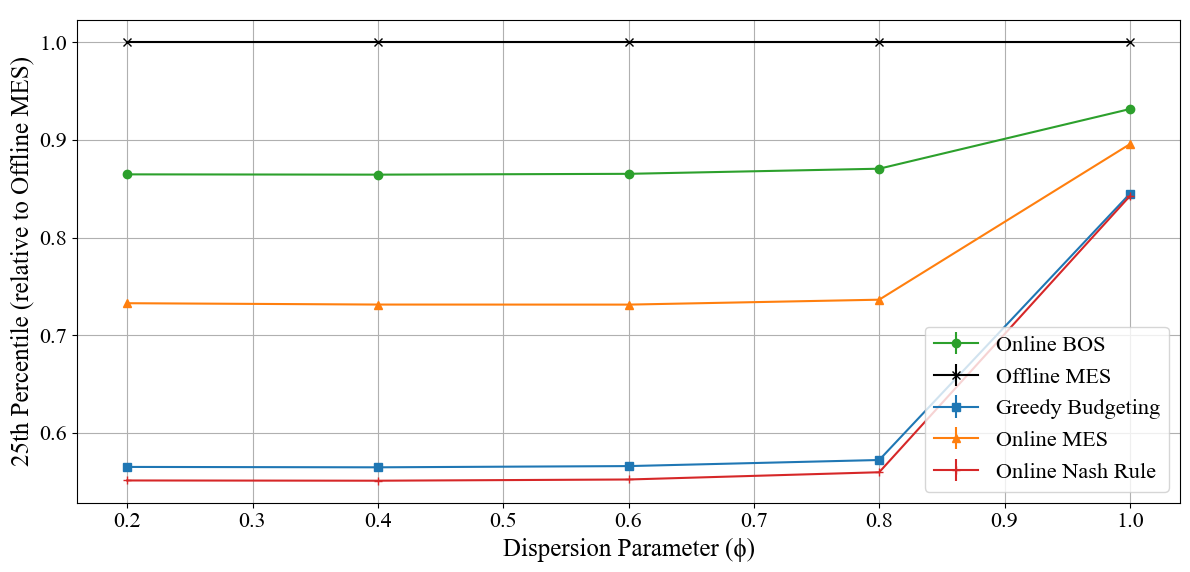}
    \end{subfigure}
    \caption{Gini coefficient against the number of voters (top-left), average satisfaction (top-right) and 25th percentile (bottom-left) relative to Offline MES against the number of voters, and 25th percentile relative to Offline MES against the dispersion parameter $\phi$ (bottom-right)
    in the Mallows model (Experiment 3).}
\label{fig:3b/5}
\end{figure}

\begin{figure}[h!]
    \centering
    \begin{subfigure}{0.496\textwidth}
        \centering
        \includegraphics[width=\textwidth]{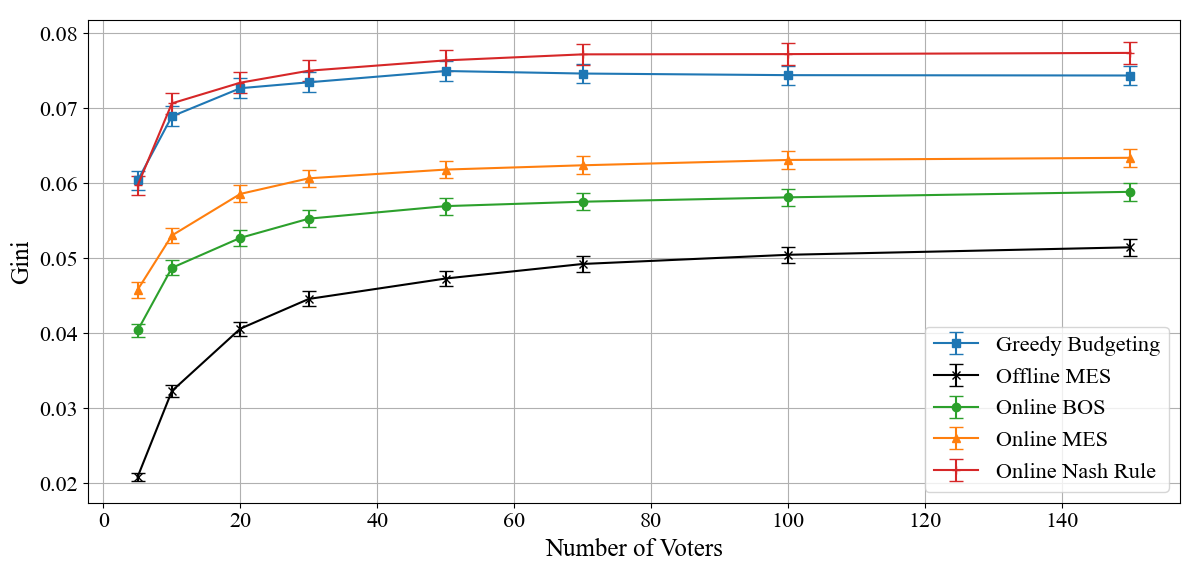}
    \end{subfigure}
    \hfill
    \begin{subfigure}{0.496\textwidth}
        \centering
        \includegraphics[width=\textwidth]{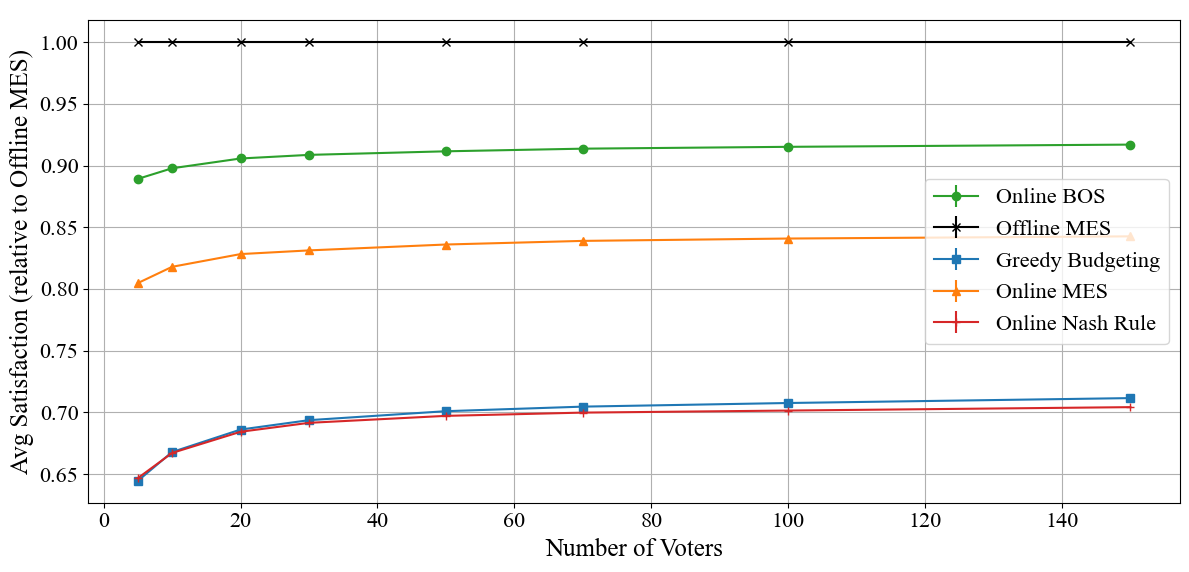}
    \end{subfigure}
    
    \vspace{0.5cm}
    
    \begin{subfigure}{0.496\textwidth}
        \centering
        \includegraphics[width=\textwidth]{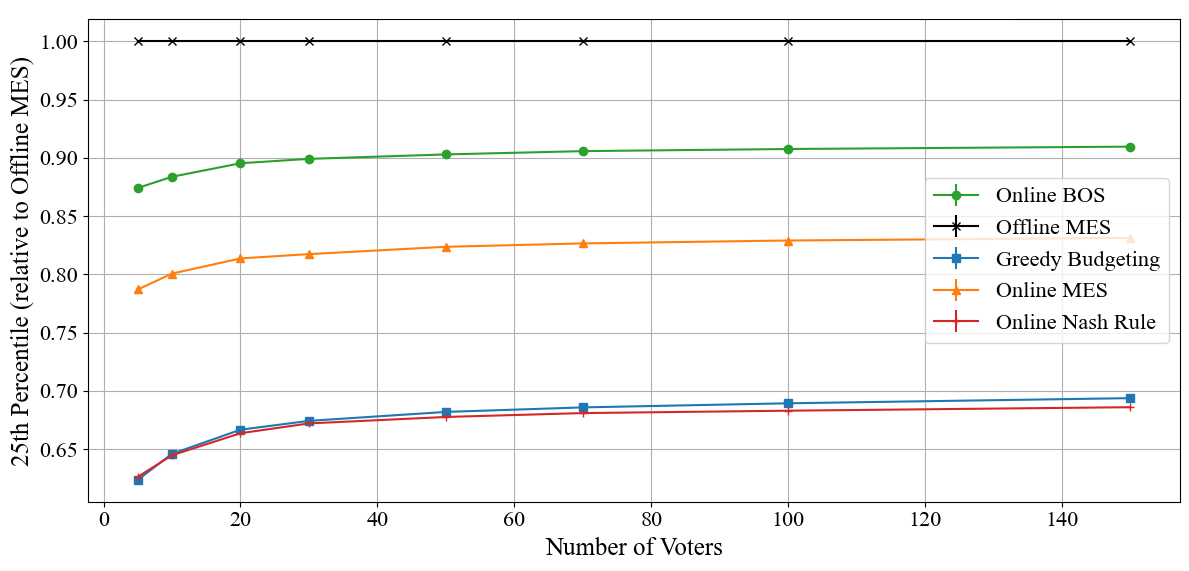}
    \end{subfigure}
    \hfill
    \begin{subfigure}{0.496\textwidth}
        \centering
        \includegraphics[width=\textwidth]{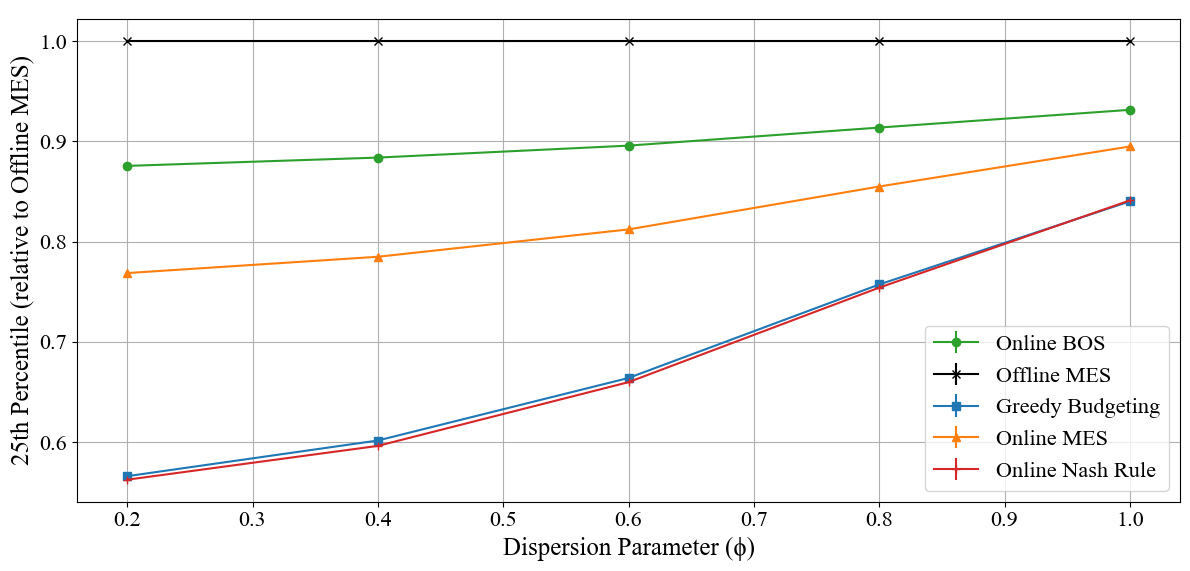}
    \end{subfigure}
    
    \caption{Gini coefficient against the number of voters (top-left), average satisfaction (top-right) and 25th percentile (bottom-left) relative to Offline MES against the number of voters, and 25th percentile relative to Offline MES against the dispersion parameter $\phi$ (bottom-right)
    in the Normalized Mallows model (Experiment 3).}\label{fig:3c/5}
\end{figure}

\subsection[Polarized Instances]{Experiment 4: Polarized Instances}
In this experiment, we generate profiles where proportionality is intuitively easy to capture. Specifically, there are two voter groups: group A approves the first half of candidates, and group B approves some from the second half. For a given value $x$, group A constitutes an $x$-percentage of the voters. Proportionality can be interpreted as requiring at least an $x$-percentage of the winning committee to come from the first half of the candidates.
We created $300$ profiles by randomly selecting values for $n$, $m$, $k$, $x$ and the approval rate of voters from group B, running $10$ different random candidate orders for each profile. 
We measure the frequency with which group A receives less representation than they deserve and we also quantify the extent of this underrepresentation.
Since JR guarantees no shortfall here, Greedy Budgeting produces outcomes with no violations by design. Online MES fails in about a fifth of the cases, Online BOS in slightly more, and the Nash Rule only in a very few. However, in terms of average shortfall, the Nash Rule shows a deficit of around 4 candidates, while Online BOS stays below 1; Online MES falls in between.
BOS also performs twice as well in terms of the maxima of deficits observed.
We report results on percentage and deficit of underperformance in \Cref{tab:underperformance_deficits}.

\begin{table}[t]
\centering
\small
\begin{tabular}{lccc}
\toprule
\textbf{Algorithm} & \textbf{Underperformance} & \textbf{Average Deficit} & \textbf{Maximum Deficit} \\
\midrule
\text{Greedy Budgeting}     & 0.00\%  & 0.0000 & 0.0000 \\
\text{Online MES}           & 19.67\% & 1.8576 & 8.9000 \\
\text{Online BOS}           & 27.00\% & 0.8556 & 4.1000 \\
\text{Online Nash Rule}     & 2.67\%  & 3.9375 & 8.1000 \\
\bottomrule
\end{tabular}
\caption{Underperformance percentages, average deficits across all instances, and maximum deficits per instance, computed as the average of maxima over multiple runs of the same instance under different candidate arrival orders (Experiment 4).}
\label{tab:underperformance_deficits}
\end{table}

\section{Concluding Discussion}
\label{sec:conclusion}
It is worth emphasizing that the proposed rules follow fundamentally different strategies, each with its own selection pattern. The Greedy Budgeting rule prioritizes candidates who appear early. This principle does not apply to Online MES (or BOS), which reserve approximately the first 37\% of candidates as an observation phase and only begin selecting afterwards.
The Nash Rule, in contrast, ensures a more distributed selection by dividing candidates into segments and selecting one from each.
These behaviors are clearly illustrated in \Cref{example1}, where the blue committee reflects the outcome of the Greedy rule, the green committee corresponds to Online MES (and BOS), and the red one to the Nash Rule. 
These differences showcase the importance of the arrival order of candidates. If some additional information were available---say, from past elections or recruitment cycles---it could guide the choice of which rule to apply. This points to a promising direction for future work: designing fair online multi-winner election rules that make use of predictions such as forecasts of future preferences. 
In fact, the negative results from \Cref{sec:certainty} might have been avoided if such predictive information had been available. Another way to potentially avoid such strongly negative findings is to significantly restrict the allowed scores in range voting.

Among the proposed rules, Greedy Budgeting has a remarkable advantage: it is the only one that does not require knowing the total number of candidates in advance (except during its final safeguard step ensuring exactly $k$ selections). The rest rely on the value of $m$ to determine the transition point between the observation and selection phases. While assuming knowledge of the number of candidates is standard in the literature of online algorithms, designing fair rules that operate without this assumption remains an important direction for future work. 
The simplicity of Greedy Budgeting is also a practical benefit and this relates to a further drawback of Online MES and BOS: their running time. Even though they run in polynomial time, they need to call the MES or BOS subroutine $O(m)$ times, which can be impractical in some scenarios. Indicatively, in the real-world datasets of Experiment 1, Greedy Budgeting rule took on average 0.01s per instance, making it significantly faster than Online MES, which needed around 1.1s. Online MES was itself about 6 times faster than Online BOS. Nash Product rule, with an average of 0.3s per instance, was also sufficiently fast.

As mentioned, our work can be seen as a follow-up to the work by \citet{Do}. There are two clear directions to generalize their work: one is what we pursued, i.e., generalizing ballots from approval to cardinal (a setting perfectly suited for scenarios such as secretary hiring that motivated our study), and the other is moving beyond multi-winner elections. The natural, yet non-trivial, next step is to explore proportionality in settings like PB, where candidates have costs and must be selected within a budget, still in an online manner. This can, in principle, be applied in online blockchain-based public funding platforms, such as Gitcoin, the Web3 Foundation, or Project Catalyst.
Participatory Budgeting brings new challenges because the number of winners is not fixed, which calls for significantly different insights or techniques.

\vspace{1.5cm}\noindent\textbf{Acknowledgments.} 
\small{
We thank Tomasz Wąs, Piotr Faliszewski and Piotr Skowron for their assistance and their feedback on earlier versions of this work. The authors were supported by the European Union (ERC, PRO-DEMOCRATIC, 101076570). Views and opinions expressed are however those of the authors only and do not necessarily reflect those of the European Union or the European Research Council. Neither the European Union nor the granting authority can be held responsible for them.}
\begin{figure}[h]
    \centering
    \includegraphics[width=0.37\linewidth]{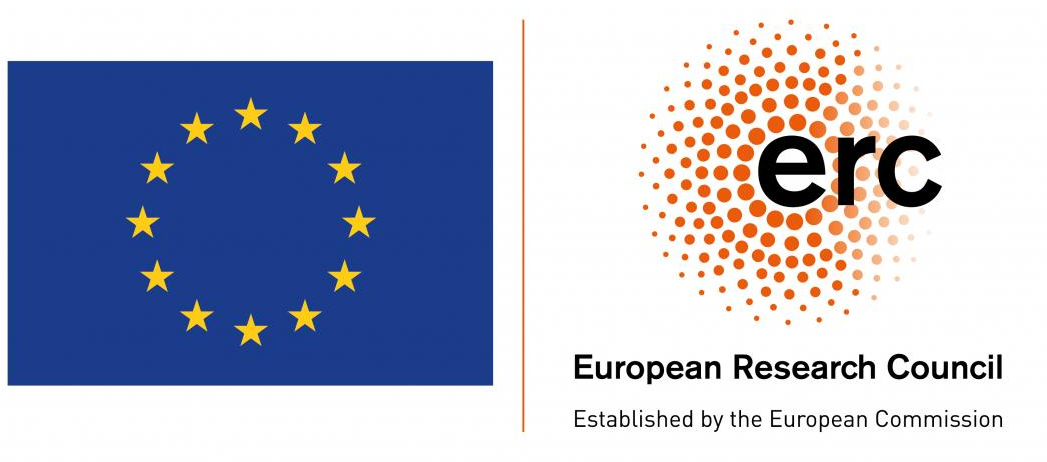}
    \label{erceu}
\end{figure}

{\small{\bibliography{biblio}}}

\end{document}